\documentclass[10pt]{amsart}
\usepackage{amsmath,amsfonts,amssymb,mathrsfs}
\usepackage{graphicx,extpfeil}
\usepackage{epsfig}
\usepackage{indentfirst, latexsym, enumerate}
\usepackage{float}
\usepackage{epsfig,subfigure}

\usepackage{extpfeil}
\usepackage{colortbl}
\usepackage{caption}
\usepackage{bm}

\title[Constants of motions for the Lohe type models with frustration]{Constants of motion for the finite-dimensional Lohe type models with frustration and applications to emergent dynamics}

\author[S.-Y. Ha]{Seung-Yeal Ha}
\address[S.-Y. Ha]{\newline Department of Mathematical Sciences and Research Institute of Mathematics \newline Seoul National University, Seoul 08826, Republic of Korea and \newline
Korea Institute for Advanced Study, Hoegiro 85, Seoul, 02455, Republic of 
Korea} \email{syha@snu.ac.kr}

\author[D. Kim]{Dohyun Kim}
\address[Dohyun Kim]{\newline Research Institute of Basic Sciences, \newline Seoul National University, Seoul 08826, Republic of Korea}
\email{dohyunkim@snu.ac.kr}

\author[H. Park]{Hansol Park}
\address[H. Park]{\newline Department of Mathematical Sciences \newline  Seoul National University, Seoul 08826, Republic of Korea}
\email{hansol960612@snu.ac.kr}

\author[S. W. Ryoo]{Sang Woo Ryoo}
\address[S. W. Ryoo]{\newline Mathematics Department, Princeton University, 
\newline Princeton, New Jersey 08544-1000, United States} \email{sryoo@math.princeton.edu}

\newtheorem{theorem}{Theorem}[section]
\newtheorem{lemma}{Lemma}[section]
\newtheorem{corollary}{Corollary}[section]
\newtheorem{proposition}{Proposition}[section]
\newtheorem{example}{Example}[section]
\newtheorem{remark}{Remark}[section]

\newtheorem{definition}{Definition}[section]

\newcommand{\bbr}{\mathbb R}
\newcommand{\bbu}{\mathbb U}
\newcommand{\bbc}{\mathbb C}
\newcommand{\bbs}{\mathbb S}

\newcommand{\bbp} {\mathbb P}
\newcommand{\bbt} {\mathbb T}

\newcommand{\bbn} {\mathbb N}

\newcommand{\mi}{\mathrm i}
\newcommand{\dg}{\dagger}
\newcommand{\kp}{\kappa}
\newcommand{\veps}{\varepsilon}

\def\charf {\textup{{\text 1}\kern-.30em {\text l}}}
\begin{document}
%%%%%%%%%%%%%%%%

\date{\today}


\thanks{\textbf{Acknowledgment.}
The work of S.-Y. Ha was partially supported by the National Research Foundation of Korea Grant (NRF-2017R1A2B2001864) funded by the Korea Government, and the work of D. Kim was supported by the National Institute for Mathematical Sciences (NIMS) grant funded by the Korea government (MSIT) (No.B19610000).}

%\thanks{\textbf{Acknowledgment.} The work of S.-Y. Ha was supported by the Samsung Science and
%Technology Foundation under Project Number SSTF-BA1401-03.}

\begin{abstract}
We present constants of motion for the finite-dimensional Lohe type aggregation models with frustration and we apply them to analyze the emergence of collective behaviors. The Lohe type models have been proposed as possible non-abelian and higher-dimensional generalizations of the Kuramoto model, which is a prototype phase model for synchronization. The aim of this paper is to study the emergent collective dynamics of these models under the effect of (interaction) frustration, which generalizes phase-shift frustrations in the Kuramoto model. To this end, we present constants of motion, i.e., conserved quantities along the flow generated by the models under consideration, and, from the perspective of the low-dimensional dynamics thus so obtained, derive several results concerning the emergent asymptotic patterns of the Kuramoto and Lohe sphere models.
\end{abstract}
\maketitle \centerline{\date}

%\tableofcontents

%%%%%%%%%%%%%%%%%%%%%%%%%%%%%%%%%%%%%%%%%%%%%%%%%%%%%%%55
\section{Introduction} \label{sec:1}
\setcounter{equation}{0}
Collective behaviors of many-body complex systems have been extensively studied in biological and physical systems, e.g., flocks of birds, swarm of bacteria, herding of sheep, arrays of Josephson junctions, etc. \cite{A-B, A-R, B-B, Ki, Ku1, Ku2, Lo-1, Lo-2, O, Pe,P-R, S-W, St,Wi2, Wi1}. However, despite their ubiquitous presence, systematic research based on rigorous mathematical modeling  was only begun  half a century ago by Winfree \cite{Wi2, Wi1} and Kuramoto \cite{Ku1,Ku2}.  In this paper, our main interest lies in generalizations of the Kuramoto model, namely the Lohe matrix and sphere models which correspond to non-abelian and higher-dimensional generalizations of the Kuramoto model, respectively. For concreteness, we begin with a description of the Kuramoto model and the Lohe type models as below.

The Kuramoto model \cite{Ku1,Ku2} describes the dynamics of a collection $\{\theta_j\}_{j=1}^N$ of $2\pi$-periodic variables, where the dynamics of $\theta_j$, the $j$-th Kuramoto oscillator, is given as follows:
\begin{equation}
\dot\theta_j = \nu_j + \frac{\kp}{N}\sum_{k=1}^N \sin(\theta_k - \theta_j), \quad j=1,\cdots,N.
\end{equation}
Here $\nu_j$ is the natural frequency of the $j$-th Kuramoto oscillator, and $\kappa>0$ denotes a positive coupling strength.

The Lohe matrix model \cite{H-R, Lo-1, Lo-2} describes an analogous dynamics on the unitary group. Here, and in the rest of the paper, we denote by $\bbu(d)$ the unitary group consisting of $d \times d$ unitary matrices. With $U_j$ and $U_j^*$, $j=1,\cdots,N$, denoting  time-dependent $d\times d$ unitary matrices and their hermitian conjugates, respectively,  let $H_j$ be a constant $d\times d$ hermitian matrix, and let $\kappa>0$ denote a positive constant denoting the interaction strength. Then, the Lohe matrix model reads as follows:
\begin{equation} \label{A-0}
\mi \dot U_j U_j ^* = H_j + \frac{\mi \kappa}{2N} \sum_{k=1}^N ( U_kU_j^* - U_j U_k^* ), \quad j=1,\cdots, N.
\end{equation}

On the other hand, the Lohe sphere model  \cite{C-C-H, C-H5, C-H4, C-H1, C-H3, C-H2,   Lo-1} describes an analogous dynamics on the $d$-dimensional unit sphere $\bbs^{d}$. Specifically, the dynamics of an ensemble $\{x_j\}_{j=1}^N$ of points on the unit sphere $\bbs^{d}$ is given as follows:
\begin{equation} \label{A-0-1}
\|x_j\|^2 \dot x_j = \Omega_j x_j + \frac{\kp}{N} \sum_{k=1}^N  (\|x_j\|^2x_k - \langle x_j,x_k\rangle x_j), \quad j=1,\cdots,N,
\end{equation}
where $\Omega_j$ is now a skew-symmetric matrix, i.e., $\Omega_j^t = -\Omega_j$, and $\kappa>0$ denotes a positive coupling strength as usual. For the case with identical oscillators, i.e., $\Omega_j = \Omega,~j= 1, \cdots, N$, system \eqref{A-0-1} has already been introduced in \cite{O} as a ``swarm on sphere" model. We remark that the Lohe sphere model \eqref{A-0-1} can be derived from the Lohe matrix model for the 3-sphere $\bbs^3$ using a special parameterization of \textcolor{black}{$\mathrm{S}\bbu(2)$}.

In \cite{Da, S-K, Zh}, frustration was introduced into the Kuramoto phase model for  more realistic modeling, and its emergent dynamics has also been extensively studied in \cite{H-K-L,H-K-L1, L-H}. The model is given as follows. Let $\theta_j$ be the phase of the $j$-th Kuramoto oscillator. Then, the Kuramoto-Sakaguchi model with uniform frustration $\alpha$ reads as 
\begin{equation} \label{A-0-2}
\dot\theta_j = \nu_j + \frac{\kp}{N}\sum_{k=1}^N \sin(\theta_k - \theta_j+ \alpha), \quad j=1,\cdots,N,
\end{equation}
where $\nu_j$ is the natural frequency of the $j$-th Kuramoto oscillator. The simple presence of $\alpha$ in the sinusoidal coupling of \eqref{A-0-2} makes the analysis harder than the original Kuramoto model with $\alpha = 0$. For example, the total phase $\sum_{j=1}^{N} \theta_j$ is not a conserved quantity any more, and the gradient flow structure is destroyed. Thus, energy estimates based on the conservation of total phase do not work in the present context. However, in spite  of the lack of conserved quantities and good structural property, the Lyapunov functional approach does still work for system \eqref{A-0-2}, as we will see below.

In the seminal paper \cite{S-W}, various constants of motion were found for the Kuramoto model \eqref{A-0-2} with frustration in the case of identical oscillators $\nu_j=\nu$, and the dynamics of this model were shown to be highly degenerate. The emergent dynamics has been classified, excluding measure zero initial data. Recently, analogous constants of motion were found for the higher-dimensional Lohe matrix model and Lohe sphere model in \cite{Lo-5}, while in \cite{C-E-M,Lo-6} the constants of motion were exploited to discover a larger class of synchronization models. 

\subsection{Key questions to be addressed} In light of the classical Lohe matrix and sphere models, the introduction of frustration into the Kuramoto model, and the discovery of constants of motion, we will address the following questions throughout this paper: \newline
\begin{itemize}
\item
(Q1):~(Derivation of the Lohe sphere and matrix models with frustration):~What will be the analogs of the Lohe sphere and Lohe matrix models with frustrated interactions?
\vspace{0.2cm}
\item
(Q2)~(Existence of nontrivial constants of motion):~Are there any nontrivial conserved quantities for Lohe type aggregation models with frustration?  
\vspace{0.2cm}
\item
(Q3):~(Application of constants of motion):~If such constants of motion exist, can we use this constant of motion in the study of large-time behaviors of the Lohe type models?
\end{itemize}
\vspace{0.2cm}
In the absence of frustration, the above questions have been extensively studied in a series of papers \cite{A-R, B-C-M, C-E-M,C-H-J-K, C-S, D-X, D-B, D-B0, D-B1,H-K-R, H-L-X, H-R, J-M-B,M-S3,M-S2,Lo-6,  M-S1, V-M1, V-M2}. Some of the results here overlap with the recent paper \cite{Lo-5}, which was published during the production of this paper, but we are including those results to provide a different perspective on the subject matter.

\subsection{Outline of main results} In what follows, we briefly discuss our main results.

First, we consider an ensemble of identical Kuramoto oscillators with frustration $\alpha$. In this case, the phase $\theta_i$ satisfies the ordinary differential equation
\begin{equation} \label{A-0-3}
\dot \theta_j = \frac{\kappa}{N} \sum_{k=1}^N \sin(\theta_k-\theta_j + \alpha), \quad j=1,\cdots,N.
\end{equation}
We present two time-invariant functionals for this case. For $\alpha \in (-\frac{\pi}{2}, \frac{\pi}{2})$, we introduce the functional $\mathcal J_\alpha(\Theta)$:
\[  \mathcal J_\alpha(\Theta) := \prod_{i=1}^N \sin \left(\frac{\theta_{i+1}-\theta_i}{2} \right) e^{\tan\alpha \sum_{i=1}^N \theta_i}  \] 
is time-invariant under the flow \eqref{A-0-3} (see Theorem \ref{T3.2}). Moreover, as a corollary, we can see that depending on the sign of $\alpha$, the Kuramoto order parameter $R := \Big| \frac{1}{N} \sum_{j=1}^{N} \theta_j \Big|$ tends to $1$ or $0$ (see Corollary \ref{C3.1}). 
For a phase configuration $\Theta$ with $\theta_i \not \equiv \theta_j\mod 2\pi, ~~ 1 \leq i, j \leq N,$ we set the functional $\mathcal K_{abcd}(\Theta)$:
\[ \mathcal K_{abcd}(\Theta) := \frac{ \Delta\theta_{ab} \cdot\Delta\theta_{cd}}{\Delta\theta_{ac} \cdot\Delta\theta_{bd}}\quad \textup{where}\quad \Delta \theta_{ab} := \sin \Big( \frac{\theta_a-\theta_b}{2} \Big ). \]
Then, the functional ${ \mathcal K}_{abcd}(\Theta)$ is time-invariant under the flow \eqref{A-0-3} (see Theorem \ref{T3.3}), even for the critical case $|\alpha|  = \frac{\pi}{2}$. We also present a low-dimensional dynamics for \eqref{A-0-3}. For a phase $\Theta$ with a configuration:
\[
\theta_j(t)=\theta_N(t), \quad  j=N-m+1,\cdots, N,\qquad \theta_j(t)\neq \theta_N(t), \quad j=1,\cdots, N-m,
\]
we introduce the auxiliary variables:
\[
x_j(t) :=\frac{1+\cos (\theta_j(t)-\theta_N(t))}{\sin (\theta_j(t)-\theta_N(t))}, \quad j=1,\cdots, N-m.
\]
Then, the dynamics of $\{x_j \}_{j=1}^{N-m}$ is fully governed by a system determined by two bounded functions ${\mathcal A}$ and ${\mathcal B}$ (see Proposition \ref{P3.1}):
\[
\begin{cases}
\displaystyle \dot{x}_j=\mathcal{A}+\mathcal{B}x_j, \quad t > 0, \\
\displaystyle x_j(0)=\frac{1+\cos(\theta_j^0-\theta_N^0)}{\sin(\theta_j^0-\theta_N^0)},\quad j = 1,\cdots, N-m.
\end{cases}
\]

\vspace{0.2cm}

Second, we present the Lohe sphere model on $\bbs^{d}$ with frustration matrix $V$ and identical matrix $\Omega_j  = \Omega$:
\begin{equation} \label{A-2}
\dot x_j = \Omega x_j  + \frac{\kp}{N}\sum_{k=1}^N \left( Vx_k-\langle x_j,Vx_k\rangle x_j \right), \quad j=1,\cdots,N.
\end{equation}
Here we employ a frustration matrix of the form $V= aI_{d+1} + W$ where $a>0$ is a positive constant, where $I_{d+1}$ denotes the  $(d+1)\times (d+1)$ identity matrix and $W$ denotes a $(d+1)\times (d+1)$ skew-symmetric matrix. For the special case with $a = 1$ and $W = 0$, system \eqref{A-2} reduces to the Lohe sphere model \eqref{A-0-1} whose emergent dynamics has been extensively studied in the previous literature. For system \eqref{A-2}, we introduce the constant of motion
\[
{\mathcal H}_{abcd}(\mathcal X):= \frac{\|x_a-x_b\|\cdot\|x_c-x_d\|}{\|x_a-x_c\|\cdot\|x_b-x_d\|}, \quad 1\leq a,b,c,d\leq N.
\]
This functional ${\mathcal H}_{abcd}(\mathcal X)$ is shown to be time-invariant under the flow \eqref{A-2} in Theorem \ref{T5.1}. It is easy to see that for identical matrices $\Omega_j = \Omega$, particles will aggregate to opposite poles ${\mathcal N}$ and ${\mathcal S}$. With the invariance of the functional ${\mathcal H}_{abcd}(\mathcal X)$ in mind, we can in fact say more: we can show that there are only two possible asymptotic patterns up to rotation (Corollary \ref{C5.1}): 
\[ (|{\mathcal N}|, |{\mathcal S}|) = (N, 0),~(N-1, 1),\]
where $|\mathcal N|$ and $|\mathcal S|$ denote number of particles which tend to $\mathcal N$ and $\mathcal S$, respectively. Moreover, we can show that there will be no periodic solution using the monotonicity of the total diameter (Corollary \ref{C5.3}). We can also show that the circles form invariant sets(see Corollary \ref{C5.2}), thanks to the classical Ptolemy's theorem, and more generally that affine subsets are preserved(see Proposition \ref{P5.3.1}).  We also provide a sufficient framework leading to complete aggregation (see Theorem \ref{T5.2}). For a spatial configuration ${\mathcal X}$ and some $m = 1, \cdots, N$,
\[ x_j \neq x_N,~ j =1,\cdots,N-m,\qquad  x_j=x_N,~j=N-m+1,\cdots,N, \]
we introduce new auxiliary variable:
\[
y_j:=x_N+\frac{2}{\|x_j-x_N\|^2}(x_j-x_N), \quad j=1,\cdots,N-m.
\]
Then, the dynamics of $\{ y_j \}_{j = 1}^{N-m}$ is governed by the three quantities $M(t)\in O(d+1)$, $a(t)>0$, \textcolor{black}{$b(t)\in \bbp_{x_N}^\perp:=\{y\in\bbr^d:\langle y, x_N\rangle=0\}$}:
\begin{equation}\label{data}
\begin{cases}
y_i(t)=M(t)(a(t)y_i(0)+b(t)),\quad i=1,\cdots,N-1,\\
x_N(t)=M(t)x_N^0, \\
M(0)=I_{d+1},\quad a(0)=1,\quad b(0)=0\in\bbp_{x_N^0}^\perp.
\end{cases}
\end{equation}
Here, three quantities $M(t)\in O(d+1)$, $a(t)>0$, $b(t)\in \bbp_{x_N^0}^\perp$ are determined by an ODE system (see Proposition \ref{P5.2}):
\[
\begin{cases}
\displaystyle a'(t)=\frac{\kappa}{N}\Bigg[1+\sum_{k=1}^{N-1}\frac{-1+\|a(t)y_k(0)+b(t)\|^2}{1+\|a(t)y_k(0)+b(t)\|^2}\Bigg]a(t),\\
\displaystyle b'(t)= \kappa b(t)+\frac{\kappa}{N}\sum_{k=1}^{N-1}\frac{2}{1+\|a(t)y_k(0)+b(t)\|^2}a(t)y_k(0),\\
\displaystyle M'(t)=M(t)L(t),\\
a(0)=1, \quad b(0)=0\in \bbp_{x_N^0}^\perp, \quad M(0)=I_{d+1}.
\end{cases}
\]

Last but not least, we present  the Lohe matrix model for identical hamiltonians with frustration:
\begin{equation} \label{NN-22}
\mi \dot U_j U_j^* = H + \frac{\mi \kappa}{2N} \sum_{k=1}^N ( VU_kU_j^* - U_iU_k^* V^*), \quad j=1,\cdots, N,
\end{equation}
where the frustration matrix  $V=aI_d+W$ is the constant $d\times d$ matrix, $I_d$ is the $d\times d$ identity matrix and $W$ is a $d\times d$ skew-symmetric matrix. It is worthwhile to mention that such a frustration operator for the Lohe matrix model (in fact, generalized Lohe matrix model proposed in \cite{H-K-R2}) was first introduced in \cite{D} and the linearization of the model around the fixed point was provided, whereas stability analysis was not yet performed. On the other hand, the constants of motion for the Lohe matrix model have been recently obtained in \cite{Lo-6} (see Section \ref{sec:6.1}). A sufficient framework for the emergent dynamics of \eqref{NN-22} has been studied in Theorem \ref{T6.1} in terms of the frustration matrix $V$ and initial data $\{ U_j^0 \}$ (see Theorem \ref{T6.1}). We also provide some class of equilibrium states using group representation (Theorem \ref{T6.2}).

\subsection{Structure of the rest of the paper} In Section \ref{sec:2}, we introduce the Lohe matrix model with frustration, which generalizes the Kuramoto model with frustration \eqref{A-0-2}, and study some basic properties.  In Section \ref{sec:3}, we present constants of motions of the Kuramoto model with frustration from a different perspective from the previous literature  \cite{M-M-S,S-W} and study low-dimensional dynamics which  is fully governed by two auxiliary  functions.  In Section \ref{sec:5}, we study constants of motion  of the Lohe sphere model, nontrivial existence of periodic solutions and low-dimensional dynamics. In Section \ref{sec:6}, we present a sufficient framework leading to the complete aggregation for an ensemble of identical particles and study a class of equilibria using an elementary property of group representation. In Appendix A, we present the proof of Proposition \ref{P3.2}.

\section{Preliminaries} \label{sec:2}
\setcounter{equation}{0}
In this section, we introduce the Lohe matrix model with frustration, and present its low-dimensional reductions to the Lohe sphere model and the Kuramoto model under the effect of frustration.

\subsection{The Lohe matrix model}  \label{sec:2.1}
In this subsection, we briefly introduce the Lohe matrix model with (interaction) frustration and study its basic properties. Let  $U_j=U_j(t)$ and $U_j^*= U_j^*(t)$ be a time-dependent $d\times d$ unitary matrix and its hermitian conjugate, and let $H_j$ and $V$ be constant $d\times d$ hermitian and unitary matrices, respectively. For motivation, let us consider the issue of how to put the frustration matrix $V$ in the coupling terms in \eqref{A-0}:
\begin{equation*} \label{B-0}
 U_kU_j^* - U_j U_k^* = U_kU_j^* -  (U_kU_j^*)^*. 
 \end{equation*}
This is simply a function of $U_k U_j^*$, and, considering that we wish to obtain a system that reduces to the Kuramoto model with frustration for $d=1$(with $V=e^{\mathrm{i}\alpha})$, there are three possible places to introduce a frustration $V$ in the quadratic term $U_k U_j^*$:
\[   VU_kU_j^*, \quad  U_k V U_j^*, \quad \mbox{and} \quad U_k U_j^* V.  \]
An important property of the Lohe matrix model is its right-translation invariance, and we wish the frustration to respect this property. Then, it is easy to see that the second choice does not lead to right-translation invariance, while the first and third do(see Lemma \ref{L2.1} (2)). Without loss of generality, we choose the first choice so that the admissible coupling term with frustration matrix will be 
\[
 VU_kU_j^* -  (VU_kU_j^*)^* = VU_kU_j^*  - U_j U_k^* V^*.
\]
The other choice will lead to a parallel discussion.

In conclusion, we define the Lohe matrix model under the effect of frustration to be governed by the following Cauchy problem:
\begin{equation}
\begin{cases} \label{B-1}
 \displaystyle \mi \dot U_j U_j ^* = H_j + \frac{\mi \kappa}{2N} \sum_{k=1}^N ( VU_kU_j^* - U_j U_k^* V^*), \quad t>0, \\
 \displaystyle U_j \Big|_{t = 0} = U_j^0, \quad U_j^0 (U_j^0)^{*} = I_d, \quad  \quad j=1,\cdots, N.
\end{cases}
\end{equation}
Next, we present two properties immediately associated with the Cauchy problem \eqref{B-1}.
\begin{lemma} \label{L2.1}
The following statements hold:
\begin{enumerate}
\item (Conservation of Unitarity)
Let $\{ U_j \}$ be a solution to \eqref{B-1}. Then $U_j U_j^*$ is conserved along the Lohe flow \eqref{B-1}:
\[
U_j(t)U_j^*(t) = I_d, \quad t\geq 0, \quad j = 1, \cdots, N.
\]
\item (Right-translation invariance) 
System \eqref{B-1} is invariant under the right multiplication action by a unitary matrix, i.e., if $L\in \mathbb U(d)$ and $W_j:=U_j L$, then $W_j$ satisfies 
\begin{equation} \label{B-1-0}
\begin{cases}
& \displaystyle \mi \dot W_j W_j^* = H_j + \frac{\mi \kappa}{2N}\sum_{k=1}^N ( VW_kW_j^* - W_j W_k^* V^*), \quad t>0, \quad j=1,\cdots, N, \\
& \displaystyle W_j \Big|_{t = 0} = W_j^0 L, \quad W_j^0 (W_j^0)^{*} = I_d, \quad  \quad j=1,\cdots, N.
\end{cases}
\end{equation}
\end{enumerate}
\end{lemma}
\begin{proof}
(i)~Note that $\eqref{B-1}_1$ can be rewritten as 
\begin{equation} \label{B-1-0-0}
\dot U_j U_j^* = -\mi H_j + \frac{\kappa}{2N} \sum_{k=1}^N ( VU_kU_j^* - U_j U_k^* V^*).
\end{equation}
We take the Hermitian conjugate of \eqref{B-1-0-0} using the relations $H_i^* = H_i$ to obtain 
\begin{equation} \label{B-1-1}
U_j \dot{U}_j^* = \mi H_j + \frac{\kappa}{2N} \sum_{k=1}^N (U_j U_k^* V^* - V U_kU_j^*).
\end{equation}
Now, we add $\eqref{B-1}_1$ and \eqref{B-1-1} to see
\begin{equation*}
\frac{d}{dt} \Big( U_j(t)U_j(t)^* \Big) = 0 , \quad \textup{or equivalently,} \quad U_j(t) U_j(t)^* = U_j^0(U_j^0)^* = I_d.
\end{equation*}
(ii) For a fixed constant matrix $L \in \mathbb U(d)$, we set
\[
 W_j := U_j L, \quad j = 1, \cdots, N.
\]
Then, it is easy to see
\begin{align}
\begin{aligned} \label{B-1-1-1}
\mi \dot{W}_j W_j^* &= \mi (\dot{U}_j L)(L^* U_j^*) =  \mi \dot{U}_j U_j^*, \\
VW_kW_j^* &= V(U_kL)(L^* U_j^*) =VU_kU_j^*, \\
W_j W_k^* V^* &= (U_j L)( L^* U_k^*) V^* = U_j U_k^* V^*.
\end{aligned}
\end{align}
In $\eqref{B-1}_1$, we substitute the above relations \eqref{B-1-1-1} to get the desired estimate \eqref{B-1-0}.
\end{proof}
\begin{remark} For future reference, we note the following variant of \eqref{B-1-0-0}, which can be obtained by multiplying $U_j$ to the right hand side of \eqref{B-1-0-0}:
\begin{equation*} \label{B-1-1-2}
\dot U_j  = -\mi H_j U_j + \frac{\kappa}{2N} \sum_{k=1}^N ( VU_k - U_j U_k^* V^* U_j)  =-\mi H_j U_j + \frac{\kappa}{2N} \sum_{k=1}^N ( VU_k - U_j (V U_k)^* U_j).
\end{equation*}
\end{remark}
\subsection{Reductions to low-dimensional models}  \label{sec:2.2} In this subsection, we review the reductions of $\eqref{B-1}_1$ to lower-dimensional synchronization models such as the Lohe sphere model and the Kuramoto model. In the absence of frustration, which is $V = I_d$ in our case, these computations were first performed in  \cite{Lo-1, Lo-2}. We include the non-frustrated case $V = I_d$ for reference and to reflect historical development.

\subsubsection{From the Lohe matrix model to the Lohe sphere model} Consider the following special case of $\eqref{B-1}_1$:
\[ d = 2, \quad V= I_2 \in \bbc^{2\times 2}. \]
In this case we may use a special parametrization of $\bbu(2)$: any $2 \times 2$ unitary matrix $U_j$ can be written as a linear combination of the Pauli matrices $\{\sigma_k\}_{k=1}^3$ and $I_2$:
\begin{equation} \label{B-2}
U_j :=     e^{-\mi\theta_j} \left( \mi \sum_{k=1}^3 x_j^k \sigma_k + x_j^4 I_2   \right) = e^{-\mi\theta_j}
\begin{pmatrix}
x_j^4 + \mi x_j^1 & x_j^2 + \mi x_j^3 \\
-x_j^2 + \mi x_j^3 & x_j^4 -\mi x_j^1
\end{pmatrix},
\end{equation}
where $I_2$ and $\{\sigma_k\}_{k=1}^3$ are given by
\begin{equation*}I_2=
\begin{pmatrix}
1 & 0 \\
0& 1\\
\end{pmatrix}
,\quad \sigma_1=
\begin{pmatrix}
1 & 0 \\
0& -1\\
\end{pmatrix}
,\quad \sigma_2 = 
\begin{pmatrix}
0 & -\mi \\
\mi & 0 \\
\end{pmatrix}
,\quad \sigma_3 = 
\begin{pmatrix}
0 &1 \\
1&0\\
\end{pmatrix}.
\end{equation*}
Similarly, $H_j$ can be expressed as a linear combination of $I_2$ and $\{\sigma_k\}_{k=1}^3$:
\begin{equation} \label{B-2-1}
H_j := \sum_{k=1}^3 \omega_j^k \sigma_k  +\nu_j I_2,
\end{equation}
where $\omega_j = (\omega_j^1, \omega_j^2, \omega_j^3)$ is a real-valued vector in $\bbr^3$, and $\nu_j$ is the natural frequency which is associated with the $\mathbb U(1)$ component of $U_j$. Now, we substitute \eqref{B-2} and \eqref{B-2-1} into $\eqref{B-1}_1$ to obtain $5N$ equations for $(\theta_j, x_j^1,\cdots,x_j^4)$:
\begin{align}
\begin{aligned} \label{B-2-1-1}
&\|x_j\|^2 \dot\theta_j = \nu_j + \frac{\kp}{N}\sum_{k=1}^N \sin(\theta_k-\theta_j)\langle x_j,x_k\rangle, \quad t >0,~~1\leq j \leq N, \\
&\|x_j\|^2 \dot x_j = \Omega_j x_j + \frac{\kp}{N} \sum_{k=1}^N \cos(\theta_k-\theta_j) \Big(\|x_j\|^2x_k - \langle x_j,x_k\rangle x_j \Big).
\end{aligned}
\end{align}
Consider next the special case
\[  \theta_j \equiv 0, \quad  \nu_j \equiv 0 \quad \textup{and} \quad \|x_j \| = 1. \]
In this case, we derive the Lohe sphere model from \eqref{B-2-1-1}:
\begin{equation*} \label{B-2-2}
\dot x_j = \Omega_j x_j + \frac{\kp}{N} \sum_{k=1}^N  \Big(x_k - \langle x_j,x_k\rangle x_j \Big).
\end{equation*}

Now, let us consider the effect of frustration. Let $V \in \bbc^{2\times 2}$ be a constant matrix which has the following form: 
\begin{equation*}
V = 
\begin{pmatrix}
v_4 + \mi v_1 & v_2 + \mi v_3 \\
-v_2 + \mi v_3 & v_4 - \mi v_1 \\
\end{pmatrix}, \quad \sum_{k=1}^4 \|v_k\|^2 =1
\end{equation*}
The same calculation above which derived \eqref{B-2-1-1} from the Lohe matrix model \eqref{B-1} now gives the Lohe sphere model with frustration on $\bbs^3$:
\begin{align}
\begin{aligned} \label{B-3}
\dot x_i &= \Omega_i x_i + \frac{\kp}{N} \sum_{k=1}^N  ({\tilde V}x_k - \langle x_i, {\tilde V}x_k\rangle x_i), \quad 1\leq i \leq N, \\
{\tilde V} &= \begin{pmatrix}
v_4 & -v_3 & v_2 & v_1 \\ v_3 & v_4 & -v_1 & v_2 \\ -v_2 & v_1 & v_4 & v_3 \\ -v_1 & -v_2 & v_3 & v_4 \\
\end{pmatrix}
= v_4 I_d + \begin{pmatrix}
0 & -v_3 & v_2 & v_1 \\ v_3 & 0 & -v_1 & v_2 \\ -v_2 & v_1 &0 & v_3 \\ -v_1 & -v_2 & v_3 & 0 \\
\end{pmatrix}.
\end{aligned}
\end{align}
From \eqref{B-3}, we can formally generalize the Lohe sphere model with frustration on $\bbs^{d}$:
\begin{equation} \label{B-3-0}
\dot x_i = \Omega_i x_i + \frac{\kp}{N} \sum_{k=1}^N  ({\tilde V}x_k - \langle x_i, {\tilde V}x_k\rangle x_i), \quad 1\leq i \leq N.
\end{equation}
For the frustration matrix, we set 
\begin{equation} \label{B-3-0-0}
 V = a I_d + W, 
 \end{equation}
where the constant $a$ is positive, $I_{d+1}$ is the $(d+1)\times (d+1)$ identity matrix and $ W$ is a $(d+1)\times (d+1)$ skew-symmetric matrix. We 
further substitute \eqref{B-3-0-0} into \eqref{B-3-0} to get 
\begin{equation} \label{B-3-1-0}
\dot x_i = \Omega_i x_i + \underbrace{\frac{\kp a }{N} \sum_{k=1}^N (x_k - \langle x_i,x_k\rangle x_i)}_\textup{synchronous motion} + \underbrace{ \frac{\kp}{N} \sum_{k=1}^N ( Wx_k - \langle x_i, Wx_k\rangle x_i ). }_\textup{ periodic motion} 
\end{equation}
\textcolor{black}{We can see that the presence of frustration matrix $W$ introduces a competition between \textit{`synchronization'} and \textit{`periodic motion'}, in the following sense. The second term on the R.H.S. of \eqref{B-3-1-0}(synchronous motion) tends to bring the oscillators together. On the other hand, since $W$ is a $(d+1) \times (d+1)$ skew-symmetric matrix,  all eigenvalues of $W$ are zero or purely imaginary. Hence, we can interpret the last term on the R.H.S. of \eqref{B-3-1-0}(periodic motion), together with $\Omega_i x_i$, tries to pull the dynamics into a periodic motion.}
\subsubsection{From the Lohe matrix model to the Kuramoto model} (This is a special case of the previous discussion.) Consider the case $d=1$ in $\eqref{B-1}_1$. In this case, we use the following ansatz:
\begin{equation*}
U_j := e^{-\mi \theta_j}, \quad H_j := \nu_j \in \bbr, \quad j = 1, \cdots, N \quad \textup{and} \quad V:= e^{-\mi\alpha}.
\end{equation*}
This yields
\begin{align}
\begin{aligned} \label{B-3-1}
& \mi \dot U_j U_j^* = \dot \theta_j  \quad \textup{and} \\
& VU_kU_j^* - U_iU_k^* V^* = e^{\mi(\theta_j-\theta_k - \alpha)} - e^{\mi(-\theta_j +\theta_k + \alpha)} = 2\mi \sin{ (\theta_j-\theta_k- \alpha)}.
\end{aligned}
\end{align}
We substitute the above relations \eqref{B-3-1} into $\eqref{B-1}_1$ to derive the Kuramoto model with frustration: 
\begin{equation} \label{B-4}
\dot \theta_j = \nu_j + \frac{\kp}{N}\sum_{k=1}^N \sin{(\theta_k-\theta_j + \alpha)}.
\end{equation}
We expand the coupling term \eqref{B-4} using basic trigonometry to obtain
\begin{equation} \label{B-4-0}
\dot\theta_j = \nu_j + \underbrace{\frac{\kp\cos\alpha}{N} \sum_{k=1}^N \sin (\theta_k -\theta_j)}_\textup{synchronous motion} + \underbrace{\frac{\kp \sin\alpha}{N} \sum_{k=1}^N \cos (\theta_k-\theta_j).}_\textup{periodic motion}
\end{equation}
Similar to the discussion at the end of Section 2.2.1, we can see that the R.H.S. of \eqref{B-4-0} again involves a natural competition between \textit{`synchronization'} and \textit{`periodic motion'}. \newline

In the following  section, we study  the constants of motion of the Kuramoto model with frustration and its application to the large-time behaviors.

%%%%%%%%%%%%%%%%%%%%%%%%%%%%%%%%%%%%%%%%%%%%%%%%%%%%%%%%%%%%%%%%%%
\section{Ensemble of identical Kuramoto oscillators with frustration} \label{sec:3}
\setcounter{equation}{0}  
In this section, we study constants of motion for the Kuramoto model  with a positive frustration $\alpha \in (0, \pi)$ and identical oscillators:
\begin{equation} \label{C-0}
\dot \theta_j = \frac{\kappa}{N} \sum_{k=1}^N \sin{ (\theta_k-\theta_j + {\tilde \alpha})}, \quad 0 \leq {\tilde \alpha} \leq \pi.
\end{equation}
We will eventually see how one can reduce this to a dynamics on $\bbr^2$.
\subsection{Constant of motion} \label{sec:3.1} We take $\alpha = {\tilde \alpha} - \frac{\pi}{2}$ to see
\begin{equation} \label{C-1}
\dot \theta_j = \frac{\kappa}{N} \sum_{k=1}^N \cos{ (\theta_k-\theta_j + \alpha) }, \quad |\alpha| \leq  \frac{\pi}{2}.
\end{equation}
Throughout the paper, we call system \eqref{C-1} the cosine Kuramoto flow with frustration. In what follows, we study the following three issues.
\begin{itemize}
\item
First, we construct a time-invariant functional of the cosine-Kuramoto model \textit{without} frustration $\alpha = 0$ (see Theorem \ref{T3.1}). 
\item
Second, we extend the above time-invariant functional to the cosine-Kuramoto model \textit{with} frustration $\alpha \in\left(-\frac{\pi}{2},\frac{\pi}{2}\right)$. (see Theorem \ref{T3.2}).
\item
Finally, we construct different time-invariant functionals for the full Kuramoto model $\alpha \in\left[-\frac{\pi}{2},\frac{\pi}{2}\right]$(see Theorem \ref{T3.3}). 
\end{itemize}

For the first step, we construct, given a phase configuration $\Theta = (\theta_1, \cdots, \theta_N)$, the functional
\[ \mathcal I(\Theta) := \prod_{i=1}^N \sin \left(\frac{\theta_{i+1}-\theta_i}{2} \right). \]
Here, we use the convention that $\theta_{N+1}=\theta_1$. We show that this functional ${\mathcal I} = {\mathcal I}(\Theta)$ is preserved along \eqref{C-1} when $\alpha=0$.
\begin{theorem} \label{T3.1}
Let $\Theta$ be a solution to \eqref{C-1} with $\alpha = 0$. Then, the functional $\mathcal I(\Theta)$ is time-invariant along the flow \eqref{C-1}:
\[   \mathcal I(\Theta(t)) =  \mathcal I(\Theta^0), \quad t \geq 0.   \]
\end{theorem}
\begin{proof}
If $\theta_i^0\equiv \theta_{i+1}^0\mod 2\pi$ for some $i=1,\cdots,N$ then the uniqueness of solutions of \eqref{C-1} tells us that $\theta_i(t)\equiv \theta_{i+1}^0(t)\mod 2\pi$ for all $t\ge 0$ for that $i$. Hence $\mathcal I(\Theta(t)) = 0= \mathcal I(\Theta^0)$ in this case. Otherwise, $\theta_i\not\equiv \theta_{i+1}\mod 2\pi$ for all $i=1,\cdots,N$, and then the following computation gives the proof:
\begin{align}
\begin{aligned} \label{C-2}
\frac{d\mathcal I}{dt} &= \sum_{j=1}^N \frac{ \partial \mathcal I}{\partial \theta_j} \frac{d\theta_j}{dt} =\mathcal I \sum_{j=1}^N \dot\theta_j \left(      -\frac{  \frac{1}{2}\cos \left(  \frac{\theta_{j+1}-\theta_j }{2} \right)       }{    \sin \left(   \frac{\theta_{j+1}-\theta_j}{2}   \right)     } + \frac{  \frac{1}{2}\cos \left(  \frac{\theta_{j}-\theta_{j-1} }{2} \right)       }{    \sin \left(   \frac{\theta_{j}-\theta_{j-1}}{2}   \right) }    \right) \\
&=\frac{\mathcal I}{2} \sum_{j=1}^N \frac{ \cos \left( \frac{\theta_j-\theta_{j-1}}{2}    \right)    }{\sin \left( \frac{\theta_j-\theta_{j-1}}{2}    \right)          } (\dot\theta_j -\dot\theta_{j-1}) \\
&=\frac{\mathcal I \kappa }{2N} \sum_{j=1}^N  \frac{ \cos \left( \frac{\theta_j-\theta_{j-1}}{2}    \right)    }{\sin \left( \frac{\theta_j-\theta_{j-1}}{2}    \right)          }  \sum_{i=1}^N ( \cos (\theta_j-\theta_i) - \cos (\theta_{j-1}-\theta_i) ) \\
&=\frac{\mathcal I \kappa}{2N} \sum_{j=1}^N  \frac{ \cos \left( \frac{\theta_j-\theta_{j-1}}{2}    \right)    }{\sin \left( \frac{\theta_j-\theta_{j-1}}{2}    \right)          }  \sum_{i=1}^N 2\sin \left(   \frac{ \theta_j + \theta_{j-1} - 2\theta_i}{2}   \right) \sin \left( \frac{\theta_{j-1}-\theta_j}{2}      \right) \\
&=-\frac{\mathcal I\kappa}{2N} \sum_{j=1}^N \sum_{i=1}^N 2\cos \left(  \frac{\theta_j-\theta_{j-1}}{2}  \right) \sin\left(  \frac{\theta_j + \theta_{j-1} -2\theta_i}{2}  \right) \\ 
&=-\frac{\mathcal I\kappa}{2N}\sum_{j=1}^N\sum_{i=1}^N \left(\sin (\theta_j-\theta_i) + \sin(\theta_{j-1} -\theta_i) \right)=-\frac{\mathcal I \kappa}{N} \sum_{i,j=1}^N \sin (\theta_j-\theta_i) \\
&= 0.
\end{aligned}
\end{align}
\end{proof}
We have verified that $\mathcal{I}(\Theta)$  is a constant of motion of \eqref{C-1} when $\alpha=0$. In the presence of frustration $\alpha\in\left(-\frac{\pi}{2},\frac{pi}{2}\right)$, we define the perturbed functional
\begin{equation} \label{C-4}
 \mathcal J_\alpha(\Theta) := \mathcal I(\Theta) e^{\tan\alpha \sum_{j=1}^N \theta_j},\quad \forall\Theta.
 \end{equation}
For $\alpha = 0$, the functional $ \mathcal J_\alpha(\Theta)$ becomes the functional  ${\mathcal I}(\Theta)$. For $\alpha\in\left(-\frac{\pi}{2},\frac{\pi}{2}\right)$ we will see that the perturbation factor $e^{\tan\alpha \sum_{j=1}^N \theta_j}$ exactly cancels out the frustration effect, so that the functional $ \mathcal J_\alpha(\Theta)$ becomes a constant of motion of \eqref{C-1}.  
\begin{theorem} \label{T3.2}
Let $\Theta$ be a solution to \eqref{C-1} with $|\alpha| < \frac{\pi}{2}$. Then, the functional ${\mathcal J}_{\alpha}(\Theta)$ is time-invariant along the flow \eqref{C-4}:
\[ {\mathcal J}_{\alpha}(\Theta(t)) =  {\mathcal J}_{\alpha}(\Theta^0), \quad t \geq 0.      \]
\end{theorem}
\begin{proof}
Again, the trivial case $\theta_i^0\equiv \theta_{i+1}^0\mod 2\pi$ for some $i=1,\cdots,N$ is treated easily, so we may assume $\theta_i^0\not\equiv \theta_{i+1}^0\mod 2\pi$ for all $i=1,\cdots,N$. Similarly to \eqref{C-2}, one has 
\begin{align*}
\begin{aligned}
\frac{d \mathcal I}{dt} &= \sum_{j=1}^N \frac{\partial \mathcal I }{\partial \theta_j} \frac{d\theta_j}{dt} =\frac{\mathcal I}{2} \sum_{j=1}^N \frac{ \cos \left( \frac{\theta_j-\theta_{j-1}}{2}    \right)    }{\sin \left( \frac{\theta_j-\theta_{j-1}}{2}    \right)          } (\dot\theta_j -\dot\theta_{j-1}) = \frac{\mathcal I \kappa }{2N} \sum_{j, k=1}^N \frac{ \cos \left( \frac{\theta_j-\theta_{j-1}}{2}    \right)    }{\sin \left( \frac{\theta_j-\theta_{j-1}}{2}    \right)          }  \\
&\times \Big[ \cos\alpha \Big(\cos (\theta_k-\theta_j)-\cos(\theta_k-\theta_{j-1}) \Big) + \sin\alpha \Big (\sin(\theta_k-\theta_j) - \sin(\theta_k-\theta_{j-1}) \Big) \Big ] \\
&=\frac{\mathcal I \kappa}{N} \sum_{j,k=1}^N \frac{ \cos \left( \frac{\theta_j-\theta_{j-1}}{2}    \right)    }{\sin \left( \frac{\theta_j-\theta_{j-1}}{2}    \right)          }  \Big[ \cos\alpha \sin \left(  \frac{ 2\theta_k-\theta_j-\theta_{j-1}}{2} \right)\sin \left(  \frac{\theta_j-\theta_{j-1}}{2} \right) \\
&\hspace{4cm} +\sin\alpha \sin \left(  \frac{\theta_{j-1}-\theta_j}{2} \right)\cos \left(  \frac{2\theta_k-\theta_j-\theta_{j-1}}{2}\right)\Big] \\
&=\frac{\mathcal I \kappa}{N} \sum_{j,k=1}^N \cos \left( \frac{\theta_j-\theta_{j-1}}{2}\right) \Big[ \cos\alpha \sin \Big( \frac{2\theta_k-\theta_j-\theta_{j-1}}{2}\Big)-\sin\alpha \cos\Big( \frac{2\theta_k-\theta_j-\theta_{j-1}}{2} \Big) \Big]  \\
&=\frac{\mathcal I \kappa}{2N} \sum_{j,k=1}^N \cos\alpha (\sin(\theta_k-\theta_{j-1}) +\sin(\theta_k-\theta_j)) -\sin\alpha (\cos(\theta_k-\theta_{j-1}) + \cos(\theta_k-\theta_j)) \\
&=-\frac{\mathcal I \kappa \sin\alpha}{N} \sum_{j,k=1}^N \cos(\theta_k-\theta_j),
\end{aligned}
\end{align*}
i.e.,
\begin{equation} \label{C-5}
\frac{d {\mathcal I} }{dt} = -\frac{\mathcal I \kappa \sin\alpha}{N} \sum_{j,k=1}^N \cos(\theta_k-\theta_j).  
\end{equation}
On the other hand, it follows from \eqref{C-1} that 
\begin{equation}\label{C-6}
\frac{d}{dt} \sum_{j=1}^N \theta_j = \frac{\kappa \cos\alpha }{N} \sum_{j,k=1}^N \cos(\theta_k-\theta_j).
\end{equation}
Then, from \eqref{C-5} and \eqref{C-6}, we find 
\begin{equation*}
\frac{d}{dt} {\mathcal I }= -\mathcal I \tan\alpha \frac{d}{dt} \sum_{j=1}^N \theta_j.
\end{equation*}
Finally, we integrate the relation above to obtain 
\begin{equation*}
\log |\mathcal I| = -\tan\alpha  \sum_{i=1}^N \theta_i + C, \quad \textup{i.e.,}\quad \mathcal I(t) = \mathcal I(t_0) e^{-\tan\alpha \sum_{j=1}^N (\theta_j(t) -\theta_j(t_0))}.
\end{equation*}
This establishes the desired relation
\begin{equation*}
\mathcal J_\alpha(\Theta(t))  = \mathcal J_\alpha(\Theta(t_0)), \quad t>0.
\end{equation*}
\end{proof}
\begin{remark}
In \cite{S-W}, the authors already found that $\mathcal I(\Theta)$ is constant of motion for \eqref{C-1} with $\sin \alpha = 0$. For the interpretation of $\mathcal I(\Theta)$, we refer the reader to Appendix B of \cite{S-W}. On the other hand, to the best of authors' knowledge, the constant of motion $\mathcal J_\alpha(\Theta)$ for \eqref{C-1} with $\sin\alpha\neq 0$ has not yet been explored in the literature.
\end{remark}

\noindent In the following,  we introduce the order parameters $(R, \phi)$, measuring a sort of average synchrony of system \eqref{C-1}.
\begin{definition} \label{D3.1}
Let $\Theta = \Theta(t)$ be a solution of \eqref{C-1}. Then, the order parameters $(R, \phi)$ are defined by the implicit relations:
\begin{align*}
\begin{aligned}
& R e^{{\mathrm i} \phi} = \frac{1}{N} \sum_{j=1}^{N} e^{{\mathrm i} \theta_j},  \quad \textup{or equivalently}, \\
& R\cos\phi = \frac{1}{N} \sum_{j=1}^N \cos\theta_j, \qquad R\sin\phi = \frac{1}{N} \sum_{j=1}^N \sin\theta_j.
\end{aligned}
\end{align*}
\end{definition}
As a corollary of Theorem \ref{T3.2}, we have the asymptotic behavior of the phase vector. 
\begin{corollary} \label{C3.1} 
Let $\Theta$ be a solution to \eqref{C-1}. Then, the following dichotomy holds.
\begin{enumerate}
\item (\cite{H-K-L, H-K-L1})
If the frustration and initial data satisfy
\[  \alpha \in \Big(0, \frac{\pi}{2} \Big), \qquad \max_{1\leq i, j \leq N} | \theta_i^0 -\theta_j^0| <  2\alpha, \]
then complete phase synchronization emerges:
\[ \lim_{t \to \infty} R(t) = 1. \]
\item
If the frustration and initial data satisfy
\begin{equation} \label{C-6-0}
 \alpha \in  \Big(-\frac{\pi}{2}, 0 \Big), \quad \theta_i^0\neq \theta_j^0 \quad 1\le i\neq j\le N, 
 \end{equation} 
then  complete incoherence emerges:
\[ \lim_{t \to \infty} R(t) = 0. \]
\end{enumerate}
\end{corollary} 
\begin{proof}
\noindent (i) For notational simplicity, we set $\tilde\alpha : = \pi/2- \alpha.$ Then, \eqref{C-1} becomes
\begin{equation} \label{C-6-0-1}
\dot \theta_j = \frac\kp N \sum_{k=1}^N \sin(\theta_k -\theta_j + \tilde \alpha).
\end{equation}
In this form \eqref{C-6-0-1} of \eqref{C-1}, the proof can be found in, for instance, \cite{H-K-L, H-K-L1}. We briefly sketch the proof. For the extremal fluctuations  $\theta_M(t)$ and $\theta_m(t)$ defined as
\begin{equation*}
\theta_M (t):= \max_{1\leq i \leq N}\theta_i(t),\quad \theta_m(t):= \min_{1\leq i \leq N} \theta_i(t),\quad t>0,
\end{equation*}
we set the diameter of phase configuration 
\begin{equation*}
D(\Theta(t)):= \theta_M(t) - \theta_m(t),\quad t\geq0.
\end{equation*}
Then, we find  
\begin{align} \label{C-6-0-2}
\begin{aligned}
\frac{d}{dt} D(\Theta) &= \dot\theta_M - \dot\theta_m = \frac{\kp}{N}\sum_{k=1}^N \sin(\theta_k - \theta_M + \tilde\alpha) - \sin(\theta_k - \theta_m + \tilde \alpha) \\
& = -\frac{2\kp}{N}\sum_{k=1}^N \sin\left( \frac{\theta_M - \theta_m}{2}\right) \cos\left( \frac{2\theta_k - (\theta_M + \theta_m) + 2\tilde\alpha}{2}  \right).
\end{aligned}
\end{align}
Since $\theta_m \leq \theta_j \leq \theta_M$ for $1\leq j \leq N$,  the term in the cosine function can be estimated as
\begin{equation*}
 -\frac{D(\Theta)}{2} + \tilde\alpha  = \frac{\theta_m - \theta_M}{2} + \tilde \alpha \leq   \frac{2\theta_k - (\theta_M + \theta_m) + 2\tilde\alpha}{2}  \leq \frac{\theta_M - \theta_m}{2} + \tilde \alpha = \frac{D(\Theta)}{2} + \tilde\alpha .
\end{equation*}
Hence, if we assume that the initial data satisfy $D(\Theta^0) <\pi-2\tilde \alpha = 2\alpha$, then a standard bootstrap argument shows that
\begin{equation} \label{C-6-0-3}
D(\Theta(t)) < D(\Theta^0) < 2\alpha,\quad t>0.
\end{equation}
Due to relation \eqref{C-6-0-3}, the cosine term in \eqref{C-6-0-2} becomes positive and thus yields exponential synchronization. \newline

\noindent (ii)~Note that 
\begin{align}
\begin{aligned} \label{C-6-1}
\sum_{i=1}^N \dot\theta_i  &= \frac{\kappa}{N}\sum_{i,j=1}^N \Big(  \cos(\theta_i-\theta_j)\cos\alpha + \sin(\theta_i-\theta_j)\sin\alpha \Big) \\
&= \frac{\kappa \cos\alpha}{N}\sum_{i,j=1}^N \Big( \cos\theta_i\cos\theta_j + \sin\theta_i\sin\theta_j \Big) \\
&= \kappa NR^2\cos\alpha \geq 0.
\end{aligned}
\end{align}
Thus, the total phase is a non-decreasing function of time $t$:
\begin{equation*}
\frac{d}{dt} \sum_{i=1}^N \theta_i   \geq 0.
\end{equation*}
Next, we claim:
\[  \sup_{0 \leq t < \infty} \sum_{i=1}^N\theta_i(t) < \infty. \]
Suppose to the contrary that the total phase is unbounded. Then since $\tan \alpha<0$, for arbitrary small $\varepsilon>0$, there is a time $T_\varepsilon>0$ such that 
\begin{equation} \label{C-7}
t>T_\varepsilon \quad  \Longrightarrow  \quad \left|e^{\tan\alpha \sum_{i=1}^N \theta_i(t) }\right| < \varepsilon.
\end{equation}
On the other hand, for $t>T_\veps$ 
\begin{equation} \label{C-8}
\mathcal J_\alpha(\Theta(t)) = \mathcal I(\Theta(t)) e^{\tan\alpha \sum_{i=1}^N \theta_i(t)}, \quad t>t_0.
\end{equation}
gives, along with \eqref{C-7},
\begin{equation*}
|\mathcal J_\alpha(\Theta(t))| \leq |\mathcal I(\Theta(t))| \varepsilon \leq \veps, \quad t>T_\varepsilon,
\end{equation*}
where we used the fact that sine function is uniformly bounded above by 1 and thus that $\mathcal I(\Theta)$ is also bounded above by 1. Since $\veps>0$ is chosen to be arbitrary, one has
\begin{equation*}
\lim_{t\to\infty} \mathcal J_\alpha(\Theta(t)) = 0.
\end{equation*}
However, since the quantity $\mathcal J_\alpha(\Theta)$ is conserved along the flow(due to Theorem \ref{T3.2}) and the initial data satisfy \eqref{C-6-0}, 
\begin{equation*}
\mathcal J_\alpha(\Theta(t)) = \mathcal J_\alpha(\Theta^0) \neq0,\quad t>0. 
\end{equation*}
This leads to the desired contradiction. Hence, $\sum_{i=1}^N \theta_i(t)$ is a non-decreasing and bounded function of time $t$, hence it converges. In particular, Barbalat's lemma together with the uniform boundedness property of the sine function implies
\[  \lim_{t\to\infty} \frac{d}{dt} \sum_{i=1}^N \dot\theta_i(t) =0. \]
Then, the above relation and \eqref{C-6-1} yield
\[ \kappa NR^2\cos\alpha \to 0 \quad \textup{i.e.,} \quad R \to 0 \quad \textup{as $t \to \infty$}. \]
\end{proof}

Note that the functional $\mathcal J_\alpha$ has a singularity at $\alpha=\pi/2$ due to the factor $\tan \alpha$. To overcome this singularity, we take the following two steps:
\newline
\begin{itemize}
\item
Step A:~By Theorem \ref{T3.2} involving parameter $\alpha$, we construct a new time-invariant functional. 

\vspace{0.2cm}
\item
Step B:~For the new-time invariant functional constructed in Step A, we show that the time-invariance of this new functional does not depend on $\alpha$ through the direct calculation.
\end{itemize}
Let $N\ge 4$, and let $\Theta$ be a phase configuration with 
\[ \theta_i \not \equiv \theta_j \mod 2\pi, \quad 1 \leq i, j \leq N. \]
Then, for four distinct indices $1 \leq a, b, c, d \leq N$, we define the cross-ratio 
\[ {\mathcal K}_{abcd}(\Theta) := \frac{ \Delta\theta_{ab} \cdot \Delta\theta_{cd}}{\Delta\theta_{ac} \cdot \Delta\theta_{bd}} \quad \textup{where} \quad \Delta \theta_{ab} := \sin \Big( \frac{\theta_a-\theta_b}{2} \Big ). \]
We will show that the cross-ratios are invariant.
\begin{theorem}{\cite{S-W}} \label{T3.3}
Let $N\ge 4$. Suppose that the frustration and initial data satisfy
\[  |\alpha| \leq \frac{\pi}{2}, \quad  \theta^0_i \not = \theta_j^0, \quad 1 \leq i\neq j \leq N. \]
Then, the functional $\mathcal K(\Theta)$ is time-invariant under the flow \eqref{C-1}.
\end{theorem}
\begin{proof} Below, we consider the two cases separately:
\[  \textup{Either}~|\alpha| < \frac{\pi}{2} \quad \textup{or} \quad |\alpha| = \frac{\pi}{2}. \]

\vspace{0.2cm}

\noindent $\bullet$~Case A $~(|\alpha| < \frac{\pi}{2})$: By permuting the indices if necessary, it is enough to consider the case 
\[ (a,b,c,d) = (1,2,3,4). \]
Then, it follows from Theorem \ref{T3.2} that 
\begin{align}
\begin{aligned} \label{C-8-0}
\mathcal J_\alpha^{(1)} &:= \mathcal I(\theta_1,\theta_2,\theta_3,\cdots,\theta_N) e^{\tan \alpha \sum_{i=1}^N \theta_i}, \\
\mathcal J_\alpha^{(2)} &:= \mathcal I(\theta_1,\theta_3,\theta_2,\cdots,\theta_N) e^{\tan \alpha \sum_{i=1}^N \theta_i}.
\end{aligned}
\end{align}
are time-invariant under the flow \eqref{C-1}. On the other hand, we have
\begin{equation} \label{C-8-1}
\frac{\mathcal J_\alpha^{(1)}}{   \mathcal J_\alpha^{(2)}} = -\frac{ \Delta\theta_{12} \Delta\theta_{34}}{\Delta\theta_{13} \Delta\theta_{24}}.
\end{equation}
Now, we combine \eqref{C-8-0} and \eqref{C-8-1} to get the desired result. \newline

\noindent $\bullet$~Case B:~Using a limiting argument, it follows from Case A that ${\mathcal K}_{abcd}(\Theta)$ is time-invariant: For fixed initial data $\Theta^0$, time $t>0$ and coupling strength but varying frustration $\alpha$, we can see that the solutions $\{\Theta_\alpha(t)\}_{\alpha\in\bbr}$(the subscript now indicating dependence on frustration $\alpha$) are continuous(actually analytic) in $\alpha$. Thus for $\alpha=\pm\frac{\pi}{2}$,
\[
{\mathcal K}_{abcd}(\Theta_{\pm \frac{\pi}{2}}(t))=\lim_{\alpha'\rightarrow \pm \frac{\pi}{2}}{\mathcal K}_{abcd}(\Theta_{\alpha'}(t))=\lim_{\alpha'\rightarrow \pm \frac{\pi}{2}}{\mathcal K}_{abcd}(\Theta_{\alpha'}^0)={\mathcal K}_{abcd}(\Theta^0).
\]
Note that this argument works even though $\mathcal J_\alpha(\Theta)$ has a singularity at $\alpha = \pm\pi/2$.

It is easy to treat the case $\alpha = \pm\pi/2$ directly, using the order parameters. For reference, we treat the case $\alpha =\pi/2$, as the case $\alpha = -\pi/2$ is similar. With $\alpha = \pi/2$, system \eqref{C-1} becomes 
\begin{equation*} \label{C-9}
\dot \theta_j = \frac{\kappa}{N} \sum_{k=1}^N \cos \left(\theta_k-\theta_j + \frac{\pi}{2} \right) = \frac{\kappa}{N} \sum_{k=1}^N  \sin (\theta_j-\theta_k). 
\end{equation*}
Now, we claim:
\begin{equation} \label{C-9-1}
\frac{d}{dt} {\mathcal K}_{abcd}(\Theta) = 0.
\end{equation}
{\bf Proof of claim}: It follows from Definition \ref{D3.1} that 
\begin{equation*}
R\cos\phi = \frac{1}{N} \sum_{i=1}^N \cos\theta_i \quad \textup{and} \quad R\sin\phi = \frac{1}{N} \sum_{i=1}^N \sin\theta_i,
\end{equation*}
and rewrite the Kuramoto model as
\begin{equation} \label{C-9-2}
\dot\theta_i = \kappa R\sin(\phi-\theta_i).
\end{equation}
We use \eqref{C-9-1} and \eqref{C-9-2} to see
\[
\dot\theta_i - \dot\theta_j = \kappa R(\sin(\phi-\theta_i)-\sin(\phi-\theta_j)) = \kappa R \cos \left( \frac{2\phi-\theta_i-\theta_j}{2}\right) \sin \left( \frac{\theta_j-\theta_i}{2}\right).
\]
For simplicity, we write 
\begin{equation*}
\Delta\theta_{ab}=S_{ab} = \sin \left( \frac{\theta_a-\theta_b}{2} \right) \quad \textup{and}\quad C_{ab} = \cos \left( \frac{\theta_a-\theta_b}{2} \right).
\end{equation*}
Then we differentiate $\mathcal K(\Theta)$ to obtain
\begin{align*}
&\frac{d}{dt} {\mathcal K}_{abcd}(\Theta) = \frac{d}{dt} \frac{ S_{ab} S_{cd}}{S_{ac}S_{bd}} \\
& \hspace{0.2cm} = \frac{    \dot S_{ab} S_{cd}S_{ac}S_{bd} + S_{ab} \dot S_{cd}S_{ac}S_{bd} - S_{ab} S_{cd}\dot S_{ac}S_{bd} -S_{ab} S_{cd}S_{ac}\dot S_{bd}     }{S_{ac}^2S_{bd}^2} \\
& \hspace{0.2cm}  =\frac{ S_{ab}S_{cd}}{S_{ac}S_{bd}} \left( \frac{\dot S_{ab}}{S_{ab}} + \frac{\dot S_{cd}}{S_{cd}} -\frac{\dot S_{ac}}{S_{ac}}-\frac{\dot S_{bd}}{S_{bd}}    \right) \\
& \hspace{0.2cm}  =\frac{ S_{ab}S_{cd}}{S_{ac}S_{bd}} \left[  \frac{C_{ab}}{S_{ab}} \left( \frac{\dot\theta_a -\dot\theta_b}{2}\right) + \frac{C_{cd}}{S_{cd}} \left( \frac{\dot\theta_c -\dot\theta_d}{2}\right) -\frac{C_{ac}}{S_{ac}} \left( \frac{\dot\theta_a -\dot\theta_c}{2}\right)-\frac{C_{cd}}{S_{cd}} \left( \frac{\dot\theta_c -\dot\theta_d}{2}\right)           \right] \\
& \hspace{0.2cm}  = \frac{KRS_{ab}S_{cd}}{2S_{ac}S_{bd}} \Biggr[  \frac{C_{ab}}{S_{ab}} S_{ab} \cos \left(\frac{2\phi-\theta_a-\theta_b}{2}\right) +\frac{C_{cd}}{S_{cd}} S_{cd} \cos \left(\frac{2\phi-\theta_c-\theta_d}{2}\right)    \\
&\hspace{2.5cm}- \frac{C_{ac}}{S_{ac}} S_{ac} \cos \left(\frac{2\phi-\theta_a-\theta_c}{2}\right) \frac{C_{cd}}{S_{cd}} S_{cd} \cos \left(\frac{2\phi-\theta_c-\theta_d}{2}\right) \Biggr] \\
& \hspace{0.2cm} =\frac{KRS_{ab}S_{cd}}{4S_{ac}S_{bd}}\Big[ \cos(\phi-\theta_a) + \cos(\phi-\theta_b) + \cos(\phi-\theta_c) + \cos(\phi-\theta_d)   \\
&\hspace{2.5cm} -\cos(\phi-\theta_a) -\cos(\phi-\theta_b) -\cos(\phi-\theta_c)-\cos(\phi-\theta_d) \Big]   \\
&=0.
\end{align*}
This yields the desired result.
\end{proof}

\subsection{Reduction to low-dimensional dynamics} \label{sec:3.2} Using the constants of motions developed so far, we will now, in this subsection, rewrite the Kuramoto model with frustration into low-dimensional dynamics depending on two auxiliary inputs.\\

Recall that 
\begin{equation}\label{Ku}
\begin{cases}
\displaystyle \dot{\theta}_j=\frac{\kappa}{N}\sum_{k=1}^N\sin (\theta_k-\theta_j +\alpha), \quad j = 1, \cdots, N, \\
\theta_j(0)=\theta_j^0,
\end{cases}
\end{equation}
In what follows, since the two-oscillator case is well-understood, we assume $N\ge 3$. The previous section showed that given any four oscillators on the circle, their geometric cross-ratio is  time-invariant. Naturally, the next step of the argument would be to foliate the state space into lower-dimensional submanifolds, in order to simplify the equations and obtain a clear picture of the dynamical properties. Indeed, this was done in the seminal paper \cite{S-W}, which verified that these cross-ratios constitute a collection of $(N-3)$-invariants and thus reduced the $N$-dimensional dynamics to 3-dimensional dynamics. Roughly speaking, this is because if we know all the cross-ratios and the positions of three points, the positions of the remaining points can be determined. In more technical terms, consider a subgroup of the M\"obius transformation group which preserves the unit disc:
\[
G=\left\{ z\mapsto e^{\mathrm{i}\phi }\frac{z-\alpha}{1-\bar{\alpha}z} : \phi\in \bbr, \alpha\in \bbc, |\alpha|<1 \right\},
\]
which is in fact isomorphic to the three-dimensional Lie group $\textup{PSL}(2,\bbr)$.
The orbits of the action of $G$ on $N$-tuples of points on the Riemann sphere in general position are precisely the connected components of the level sets with respect to the collection of all possible cross-ratios. Hence, the state at time $t$ can be described as the result of some transformation $M(t)\in G$ acting on the initial state, and so the $N$-dimensional dynamics of $(\theta_1,\cdots,\theta_N)\in \bbt^N$ can be described with the 3-dimensional dynamics of $M(t)\in G$. This argument has been pursued in  \cite{C-E-M}. Below, we use a slightly different approach. Once we effect a stereographic projection with respect to $e^{\mathrm{i}\theta_N}$, the conservation of the cross-ratios of the particles on the circle is equivalent to the conservation of the cross-ratios of the images particles under the stereographic projection.
\begin{figure}[h]
\centering
\includegraphics[width=0.6\textwidth]{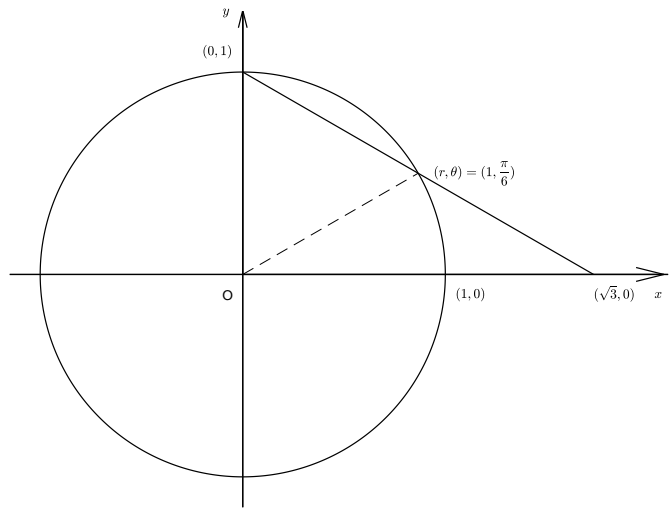} 
\caption{Stereographic projection} 
\end{figure}
Thus the image of the stereographic transformation changes along the one-dimensional affine transformation group $\textup{Aff}(1,\bbr)$, which is a 2-dimensional Lie group. The remaining 1-degree of freedom is reflected in the location of $e^{\mathrm{i}\theta_N}$.

We use this argument involving stereographic projection, instead of the M\"obius transformation group, for two reasons:
\begin{itemize}
\item 
The structure of $\textup{Aff}(1,\bbr)$ is more intuitive than that of $G \cong PSL(2,\bbr)$.
\item 
This method can easily be extended to higher-dimensional Lohe sphere models and thus can give a unified framework for understanding constants of motion.
\end{itemize}
Next, we prepare a formal setup. We rearrange the oscillators which are equal to $\theta_N$ modulo $2\pi$ to the end of the $N$-tuple and perform some modulo $2\pi$ shifts so that
\[
\theta_j^0=\theta_N^0, \quad j=N-m+1,\cdots, N,\qquad \theta_j^0\not\equiv \theta_N^0 \mod 2\pi, \quad j=1,\cdots, N-m,
\]
for some $1\le m \le N$. By the autonomy of \eqref{Ku}, we have  for all $t\ge 0$
\begin{equation} \label{C-9-3}
\theta_j(t)=\theta_N(t), \quad  j=N-m+1,\cdots, N,\qquad \theta_j(t)\neq \theta_N(t), \quad j=1,\cdots, N-m.
\end{equation}
We introduce the auxiliary variables
\begin{equation}\label{param}
\beta_j(t):=\theta_j(t)-\theta_N(t), \quad x_j :=\frac{1+\cos \beta_j}{\sin \beta_j}=\frac{\sin\beta_j}{1-\cos\beta_j},\quad j=1,\cdots,N-m.
\end{equation}
This clearly satisfies
\begin{equation}\label{param-2}
\sin\beta_j=\frac{2x_j}{x_j^2+1},\quad \cos\beta_j=\frac{x_j^2-1}{x_j^2+1}.
\end{equation}
The geometric meaning of these variables is that we rotate the unit circle $\bbs^1\subset \bbr^2$ so that $\theta_N$ is situated at the point $(1,0)$, and we use stereographic projection to project the $j$-th oscillator at $e^{\mathrm{i}\beta_j}$ onto the point $(0,x_j)$ of the $y$-axis. In the next proposition, we derive the Cauchy problem for $x_j$.
\begin{proposition}\label{P3.1}
The variables $\{x_j\}_{j=1}^{N-m}$ in \eqref{param} satisfy
\begin{equation}\label{stereo}
\begin{cases}
\displaystyle \dot{x}_j=\mathcal{A}+\mathcal{B}x_j, \quad t > 0, \\
\displaystyle x_j(0)=\frac{1+\cos(\theta_j^0-\theta_N^0)}{\sin(\theta_j^0-\theta_N^0)},\quad j = 1,\cdots, N-m,
\end{cases}
\end{equation}
where the coefficients ${\mathcal A} = \mathcal{A}(N,m,K,x_1,\cdots,x_{N-m},\alpha)$ and ${\mathcal B} = \mathcal{B}(N,m,K,x_1,\cdots,x_{N-m},\alpha)$ are explicitly given by the following relations:
\begin{align*}
\begin{aligned}
\mathcal{A} &:=\frac{\kappa}{N}\left[m\sin\alpha+\sum_{k=1}^{N-m}\left(\frac{2x_k}{x_k^2+1}\cos\alpha+\frac{x_k^2-1}{x_k^2+1}\sin\alpha\right)\right], \\
\mathcal{B} &:=\frac{\kappa}{N}\left[m\cos\alpha+\sum_{k=1}^{N-m}\left(-\frac{2x_k}{x_k^2+1}\sin\alpha+\frac{x_k^2-1}{x_k^2+1}\cos\alpha\right)\right].
\end{aligned}
\end{align*}
\end{proposition}
\begin{proof} We consider $\eqref{stereo}_2$ and $\eqref{stereo}_1$, separately.  \newline 

\noindent $\bullet$ (Derivation of $\eqref{stereo}_2$):~This follows from the relation \eqref{param} directly. \newline

\noindent $\bullet$ (Derivation of $\eqref{stereo}_1$):~From the relation  \eqref{C-9-3}, it is easy to see that for $j=1,\cdots,N-m$,
\begin{align*}
\begin{aligned}
\dot{\theta}_j &=\frac{\kappa}{N}\sum_{k=1}^N\sin (\theta_k-\theta_j +\alpha) =\frac{\kappa}{N}\left(-m\sin(\beta_j-\alpha)+\sum_{k=1}^{N-m}\sin (\beta_k-\beta_j +\alpha)\right), \\
\dot{\theta_N} &=\frac{\kappa}{N}\sum_{k=1}^N\sin (\theta_k-\theta_N+\alpha)=\frac{\kappa}{N}\left(m\sin\alpha+\sum_{k=1}^{N-m}\sin (\beta_k+\alpha)\right).
\end{aligned}
\end{align*}
Thus, for $j=1,\cdots,N-m$ one has 
\begin{align}
\begin{aligned} \label{C-9-4}
\dot{\beta}_j &=\dot{\theta}_j-\dot{\theta}_N \\
&=\frac{\kappa}{N}\left[-m\sin(\beta_j-\alpha)+\sum_{k=1}^{N-m}\sin (\beta_k-\beta_j +\alpha)-m\sin\alpha-\sum_{k=1}^{N-m}\sin (\beta_k+\alpha)\right]\\
&=\frac{\kappa}{N}\left[m\sin\alpha(\cos\beta_j-1)-m\sin\beta_i\cos\alpha+\sum_{k=1}^{N-m}\left(\sin (\beta_k+\alpha)(\cos\beta_j-1)-\cos(\beta_k+\alpha)\sin\beta_j \right)\right].
\end{aligned}
\end{align}
Again, for $j=1,\cdots,N-m$, it follows from \eqref{param} and \eqref{param-2} that
\begin{align*}
\dot{x}_j &=\frac{d}{dt}\left(\frac{\sin \beta_i}{1-\cos\beta_j}\right)=\left(\frac{\cos\beta_j \cdot(1-\cos\beta_j)-\sin\beta_j \cdot\sin\beta_j}{(1-\cos\beta_j)^2}\right)\dot{\beta}_j  \\
&=\frac{\cos\beta_j -1}{(1-\cos\beta_j)^2}\dot{\beta}_j=\frac{1}{\cos\beta_j-1}\dot{\beta}_j \\
&=\frac{\kappa}{N}\left[m\sin\alpha+\frac{m\sin\beta_j}{1-\cos\beta_j}m\cos\alpha+\sum_{k=1}^{N-m}\left(\sin (\beta_k+\alpha)+\cos(\beta_k+\alpha)\frac{\sin\beta_j}{1-\cos\beta_j}\right)\right]\\
&=\frac{\kappa}{N}\biggl[m\sin\alpha+mx_j \cos\alpha \\
& \hspace{0.5cm} +\sum_{k=1}^{N-m}\left\{\left(\frac{2x_k}{x_k^2+1}\cos\alpha+\frac{x_k^2-1}{x_k^2+1}\sin\alpha\right) +\left(-\frac{2x_k}{x_k^2+1}\sin\alpha+\frac{x_k^2-1}{x_k^2+1}\cos\alpha\right)x_j \right\} \biggl]\\
&=\frac{\kappa}{N}\left[m\sin\alpha+\sum_{k=1}^{N-m}\left(\frac{2x_k}{x_k^2+1}\cos\alpha+\frac{x_k^2-1}{x_k^2+1}\sin\alpha\right)\right]\\
&\hspace{0.5cm} +x_j \frac{\kappa}{N}\left[m\cos\alpha+\sum_{k=1}^{N-m}\left(-\frac{2x_k}{x_k^2+1}\sin\alpha+\frac{x_k^2-1}{x_k^2+1}\cos\alpha\right)\right] \\
&= \mathcal{A}+\mathcal{B}x_j.
\end{align*}
\end{proof}
\begin{lemma}\label{L3.1}
The following assertions hold. For $t \geq 0$ and $ i,j,k,l=1,\cdots,N-m$, 
\begin{align*}
\begin{aligned}
& \textup{(i)}~\frac{d}{dt}(x_i-x_j)=\mathcal{B}(x_i-x_j). \\
& \textup{(ii)}~(x_i(t)-x_j(t))(x_k^0-x_l^0)-(x_i^0-x_j^0)(x_k(t)-x_l(t))=0.
\end{aligned}
\end{align*}
\end{lemma}
\begin{proof} (i)~We use Proposition \ref{P3.1} to get 
\begin{equation} \label{C-9-4}
\frac{d}{dt}(x_i-x_j)=(\mathcal{A}+\mathcal{B}x_i)-(\mathcal{A}+\mathcal{B}x_j)=\mathcal{B}(x_i-x_j).
\end{equation}
(ii)~Note that $\mathcal{A}$ and $\mathcal{B}$ are independent of the choice of $i$ and $j$. We now solve \eqref{C-9-4} to obtain
\begin{align*}
\begin{aligned}
x_i(t)-x_j(t) &= (x_i^0-x_j^0)\cdot \exp\left[\int_0^t \mathcal{B}(s)ds\right], \\
x_k(t)-x_l(t) &= (x_k^0-x_l^0)\cdot \exp\left[\int_0^t \mathcal{B}(s)ds\right].
\end{aligned}
\end{align*}
This yields the desired relation. 
\end{proof}
\begin{remark}
Since the stereographic projection used in \eqref{param} takes the cross ratio to the ratio of side lengths on the real line, the last statement of Lemma \ref{L3.1} is equivalent to the statement that the cross ratios are constant.
\end{remark}
Thus, there is an affine transformation of the real line that takes the configuration $\{x_i(0)\}_{i=1}^{N-m}$ to each $\{x_i(t)\}_{i=1}^{N-m}$, i.e., for each time $t\ge 0$, we may find two functions $f(t)$ and $g(t)$ which satisfy
\begin{equation}\label{affine}
x_j(t)=g(t)+f(t)x_j^0,\quad j=1,\cdots,N-m.
\end{equation}
However, the choice of such $f(t)$ and $g(t)$ would not be unique, if $x_1^0=\cdots=x_{N-m}^0$. Of course, in this case, we do not need all this machinery, because this is a weighted 2-oscillator case which is easy to analyze. Nevertheless, we can choose $f(t)$ and $g(t)$ in a consistent manner. The heuristic argument runs as follows. First, clearly we should put $f(0)=1$ and $g(0)=0$, and since the arrangement of the $x_i$'s is invariant, the affine transformation should be orientation-preserving, i.e., $f(t)>0$ for all $t\ge 0$. We substitute \eqref{affine} in \eqref{stereo} to obtain
\[
g'(t)+f'(t)x_j^0 = \mathcal{A}+\mathcal{B}x_j(t) =\mathcal{A}+\mathcal{B}g(t)+\mathcal{B}f(t)x_j^0,\quad j=1,\cdots, N-m.
\]
Hence, it would be reasonable to formulate
\begin{equation} \label{C-9-4-1}
\frac{d}{dt} \begin{pmatrix}
f \\ g
\end{pmatrix}
= 
\tilde{\mathcal{B}} 
\begin{pmatrix}
f \\ g
\end{pmatrix}
+ 
\begin{pmatrix}
0 \\ 
\tilde{\mathcal{A}}
\end{pmatrix}
\quad
\textup{and} \quad
\begin{pmatrix}
f(0) \\ g(0)
\end{pmatrix}
=
\begin{pmatrix}
1 \\ 0
\end{pmatrix},
\end{equation}
where $\tilde{\mathcal{A}} = \tilde{\mathcal{A}}(N,m,\kappa,f,g,\alpha)$ and $\tilde{\mathcal{B}} = \tilde{\mathcal{B}}(N,m,\kappa,f,g,\alpha)$ are now expressed in terms of $f$ and $g$:
\begin{align}
\begin{aligned}  \label{C-9-5}
\tilde{\mathcal{A}} &=\frac{\kappa}{N}\left[m\sin\alpha+\sum_{k=1}^{N-m}\left(\frac{2fx_k^0+2g}{(fx_k^0+g)^2+1}\cos\alpha+\frac{(fx_k^0+g)^2-1}{(fx_k^0+g)^2+1}\sin\alpha\right)\right], \\
\tilde{\mathcal{B}} &=\frac{\kappa}{N}\left[m\cos\alpha+\sum_{k=1}^{N-m}\left(-\frac{2(fx_k^0+g)}{(fx_k^0+g)^2+1}\sin\alpha+\frac{(fx_k^0+g)^2-1}{(fx_k^0+g)^2+1}\cos\alpha\right)\right].
\end{aligned}
\end{align}
%\begin{align*}
%g'(t)=&\frac{K}{N}\left[m\sin\alpha+\sum_{k=1}^{N-m}\left(\frac{2f(t)x_k^0+2g(t)}{(f(t)x_k^0+g(t))^2+1}\cos\alpha+\frac{(f(t)x_k^0+g(t))^2-1}{(f(t)x_k^0+g(t))^2+1}\sin\alpha\right)\right]\\
%&+g(t)K\cos\alpha,
%\end{align*}
Formally speaking, this can be stated as follows:
\begin{proposition}\label{P3.2}
Suppose that $(f, g)$ satisfies the Cauchy problem \eqref{C-9-4-1}--\eqref{C-9-5}. Then the relation \eqref{affine} holds.
\end{proposition}
\begin{proof} 
We leave the detailed proof to Appendix A. 
\end{proof}

\begin{remark} Below, we provide several comments on the result of Proposition \ref{P3.2}. 
\begin{enumerate}
\item We reduced the Kuramoto model \eqref{Ku} for $N$ nonlinearly coupled equations to a nonlinear system of two equations for $f$ and $g$. The degree of freedom $N$ in \eqref{Ku} is given in the initial conditions; one degree of freedom is taken care of by considering the relative phase differences $\beta_i$, and then the remaining degree of freedom $N-1$ is given in the governing differential equation of Proposition \ref{P3.2}. 
\vspace{0.2cm}
\item Proposition \ref{P3.2} highlights the dynamical significance of the invariance of the cross ratios: the $N$ equations of \eqref{Ku} can be reduced to two equations.
\vspace{0.2cm}
\item We distinguished ($\tilde{\mathcal{A}}, \tilde{\mathcal{B}}$) from ($\mathcal{A}, \mathcal{B}$) in order to provide a rigorous proof of Proposition \ref{P3.2}. The proposition tells us that such a distinction is unnecessary for practical purposes.
\end{enumerate}
\end{remark}
Below, we provide explicit examples for coefficient functions ${\mathcal A}$ and ${\mathcal B}$ for $\alpha = 0$ and $\frac{\pi}{2}$, respectively. 
\begin{example}
\begin{enumerate}
\item
For the ``sine-Kuramoto model \eqref{Ku}" with $\alpha = 0$, we have
\[
\mathcal{A}=\frac{\kappa}{N}\sum_{k=1}^{N-m}\frac{2f(t)x_k^0+2g(t)}{(f(t)x_k^0+g(t))^2+1}, \quad \mathcal{B}=\frac{\kappa}{N}\left[m+\sum_{k=1}^{N-m}\frac{(f(t)x_k^0+g(t))^2-1}{(f(t)x_k^0+g(t))^2+1}\right].
\]
\item
For the ``cosine-Kuramoto model  \eqref{Ku}" with $\alpha = \frac{\pi}{2}$, we have
\[
\mathcal{A}=\frac{\kappa}{N}\left[m+\sum_{k=1}^{N-m}\frac{(f(t)x_k^0+g(t))^2-1}{(f(t)x_k^0+g(t))^2+1}\right], \quad 
\mathcal{B}=-\frac{\kappa}{N}\sum_{k=1}^{N-m}\frac{2(f(t)x_k^0+g(t))}{(f(t)x_k^0+g(t))^2+1}.
\]
\end{enumerate}
\end{example}

%%%%%%%%%%%%%%%%%%%%%%%%%%%%%%%%%%%%%%%%%%%%%%%%%%%%%%%%%%%%%%%%%%%%%%%%%%%%%%%%%%
 
%%%%%%%%%%%%%%%%%%%%%%%%%%%%%%%%%%%%%%%%%%%%%%%%%%%%%%%%%%%%%%%%%%%%%%%%%%%%%%%%%%%%%
\section{The Lohe sphere model} \label{sec:5}
\setcounter{equation}{0} 
In this section, we study the constants of motion to the Lohe sphere model and establish the non-existence of limit cycle solutions for identical oscillators. We also provide a reduction of the Lohe sphere model into a low-dimensional system. 
\subsection{Constants of motion} \label{sec:5.1}
In this section, we study the constants of motion as a generalization of the constants of motion for the Kuramoto model in Section \ref{sec:3.1}. First, we consider the Cauchy problem of the Lohe sphere model on $\bbs^d$ without frustration:
\begin{equation}  \label{Z-1}
\begin{cases}
\displaystyle \dot x_i = \Omega_i x_i + \frac{\kappa}{N} \sum_{j=1}^N (x_j -\langle x_i,x_j\rangle x_i), \quad t > 0, \\
\displaystyle x_i(0) =x_i^0,
\end{cases}
\end{equation} 
where $\kappa$ is a nonnegative constant coupling strength. In what follows, we consider the following identical particle case: 
\begin{equation} \label{Z-1-0}
\Omega_i \equiv O \quad \textup{and}\quad x_i^0\neq x_j^0 \quad \textup{for all $i\neq j \in \{1,\cdots,N\}$}.
\end{equation}
As motivation, let $x_i=(\cos\theta_i,\sin\theta_i) \in \bbs^1$ be a point on the unit circle. Then we can easily check that
\begin{equation} \label{Z-1-1}
2\left|\sin\left( \frac{\theta_i-\theta_j}{2}\right) \right| = \|x_i-x_j\|.
\end{equation}
Then, based on the constant of motion in Section \ref{sec:3.1} and \eqref{Z-1-1}, we define the following cross-ratio functional of four points in general position \eqref{Z-1-0}:
\begin{equation}\label{Z-2}
H_{abcd}(\mathcal X):= \frac{\|x_a-x_b\|\cdot\|x_c-x_d\|}{\|x_a-x_c\|\cdot\|x_b-x_d\|}, \quad 1\leq a,b,c,d\leq N.
\end{equation}

\vspace{0.2cm}

 We here mention that conservation of the functional $H$ has been independently verified in recent work \cite{Lo-5}, where  the Watanabe-Strogatz transform for the Kuramoto model is generalized to its high-dimensional model \eqref{Z-1} with $\Omega_i \equiv O$. The detailed argument can be found in \cite{Lo-5}. In the following theorem, we show the aforementioned functional \eqref{Z-2} is time-invariant under the flow \eqref{Z-1}.

\begin{theorem} \label{T5.1}
Let $N\ge 4$, and let $\{ x_i(t) \}$ be a solution to the Lohe sphere model \eqref{Z-1}--\eqref{Z-1-0}. Then for any distinct indices $a,b,c,d$, the functional $H_{abcd}(\mathcal X)$ is invariant under the Lohe flow \eqref{Z-1}--\eqref{Z-1-0}:
\begin{equation*}
\frac{d}{dt} H_{abcd}(\mathcal X) = 0. 
\end{equation*}
\end{theorem}
\begin{proof}
For notational simplicity, we set the distance between two positions $x_a$ and $x_b$ and the total centroid to be the following:
\[ \ell_{ab} := \|x_a-x_b\|  \quad \textup{and} \quad \bar x := \frac{1}{N} \sum_{k=1}^N x_k.  \]
We differentiate $\ell_{ab}^2$ to obtain 
\[
\ell_{ab} \frac{d}{dt} \ell_{ab} = -\kappa (1-\langle x_a,x_b \rangle) ( \langle \bar x ,x_a \rangle + \langle \bar x,x_b\rangle).
\]
Since $\ell_{ab}^2 = 2(1-\langle x_a,x_b\rangle )$, we have 
\begin{equation*}
\frac{ \dot \ell_{ab}}{\ell_{ab}} = -\frac{\kappa}{2} \Big( \langle \bar x ,x_a \rangle + \langle \bar x,x_b\rangle \Big).
\end{equation*}
Then we finally differentiate $H_{abcd}(\mathcal X)$ to attain
\begin{align*}
\frac{d}{dt} H_{abcd}(\mathcal X) &= \frac{ \dot\ell_{ab} \ell_{cd} \ell_{ac}\ell_{bd}  + \ell_{ab} \dot \ell_{cd} \ell_{ac}\ell_{bd}   -\ell_{ab} \ell_{cd}\dot\ell_{ac}\ell_{bd} -\ell_{ab} \ell_{cd} \ell_{ac} \dot\ell_{bd}             }{        \ell_{ac}^2\ell_{bd}^2           } \\
&=\frac{ \ell_{ab} \ell_{cd}}{\ell_{ac}\ell_{bd}} \left(  \frac{\dot\ell_{ab}}{\ell_{ab}} +  \frac{\dot\ell_{cd}}{\ell_{cd}}  - \frac{\dot\ell_{ac}}{\ell_{ac}} - \frac{\dot\ell_{bd}}{\ell_{bd}}     \right) \\
&= -\frac{\kappa}{2} \frac{ \ell_{ab} \ell_{cd}}{\ell_{ac}\ell_{bd}} \\
& \times \Big( \langle \bar x, x_a\rangle + \langle \bar x, x_b\rangle + \langle \bar x, x_c\rangle+ \langle \bar x, x_d\rangle - \langle \bar x, x_a\rangle -\langle \bar x, x_c\rangle - \langle \bar x, x_b\rangle-\langle \bar x, x_d\rangle \Big) \\
&=0.
\end{align*}
This yields the desired estimate.
\end{proof}
In what follows, we present three applications of Theorem \ref{T5.1}. The first result concerns the number of particles converging toward opposite poles on the unit sphere. 
\begin{corollary} \label{C5.1}
Let $\{ x_i(t) \}$ be a solution to the Lohe system \eqref{Z-1}--\eqref{Z-1-0} such that all initial positions $x_i^0$ are different from each other, and we set $\mathcal M_N$ and ${\mathcal M}_S$ be numbers of the particles which aggregate to the north pole $\mathcal N$ and south pole $\mathcal S$, respectively. Then one has
\begin{equation*}
\min \{  \mathcal M_{\mathcal N}  ,\mathcal M_{\mathcal S}  \} \leq 1.
\end{equation*}
\end{corollary}

\begin{proof}
Let $x_i(t)$ be the position of the $i$-th particle at time $t$ which is located on the $d$-dimensional sphere. To derive a contradiction, suppose to contrary, i.e., that two distinct points $x_a$ and $x_b$ approach the north pole $\mathcal N$ and two other distinct points $x_c$ and $x_d$ approach the south pole $\mathcal S$, respectively. Then, it follows from Theorem \ref{T5.1} that 
\begin{equation} \label{Z-3}
\frac{ \ell_{ac} \ell_{bd} }{ \ell_{ab}\ell_{cd}     }(t) = \frac{ \ell_{ac} \ell_{bd} }{ \ell_{ab}\ell_{cd}     }(0) =:  {\mathcal C},\quad t\geq0.
\end{equation}
Clearly, the constant $\mathcal I$ cannot be infinite since it is determined by the initial data and the denominator of the ratio functional $H_{abcd}(\mathcal X)$ is nonzero. However, since we assume that $\ell_{ab}$ and $\ell_{cd}$ are converging to zero, the L.H.S. of \eqref{Z-3} diverges as time goes to infinity, whereas the R.H.S. of \eqref{Z-3} still remains to be positive constant. This gives a contradiction, and we obtain the desired conclusion.
\end{proof}
Our second corollary deals with the invariance of circles. For this, we briefly recall the classical Ptolmey theorem (page 308 \cite{Be}) without proof: if the four vertices of a cyclic quadrilateral are denoted as $A,B,C$ and $D$ in counterclockwise order, then lengths of the four sides and the two diagonals of the cyclic quadrilateral satisfy 
\begin{equation} \label{ptolemy} 
|AB|\cdot |CD|+ |BC|\cdot |AD |= |AC|\cdot |BD|.
\end{equation} 
Moreover, the converse of Ptolemy's theorem is also true. In other words, if the four vertices in a quadrilateral $A,B,C$ and $D$ satisfy the relation \eqref{ptolemy}, then the quadrilateral can be inscribed in a circle, that is, the four vertices lie on a circle. 

\begin{corollary} \label{C5.2}
Let $\{ x_i(t) \}$ be a solution to the Lohe sphere model \eqref{Z-1}--\eqref{Z-1-0}. Suppose that any four points lie on the same circle at $t=0$. Then these four points still remain on the same circle for all $t>0$.
\end{corollary}
\begin{proof}
Suppose that the four points $x_a,x_b,x_c$ and $x_d$ initially lie on the same circle. By Ptolemy's theorem, we have
\begin{equation*} \label{Z-3-1}
\ell_{ab}^0 \ell_{cd}^0 + \ell_{bc}^0\ell_{ad}^0 = \ell_{ac}^0 \ell_{bd}^0 \quad \textup{or equivalently}\quad  \frac{ \ell_{ab}^0 \ell_{cd}^0}{\ell_{ac}^0 \ell_{bd}^0} + \frac{ \ell_{bc}^0\ell_{ad}^0      }{\ell_{ac}^0 \ell_{bd}^0} = 1. 
\end{equation*}
Then, Theorem \ref{T5.1} implies 
\begin{equation*}
\frac{\ell_{ab} \ell_{cd}(t)}{\ell_{ac} \ell_{bd}(t)} + \frac{\ell_{bc}\ell_{ad}(t)}{\ell_{ac} \ell_{bd}(t)} =\frac{ \ell_{ab}^0 \ell_{cd}^0}{\ell_{ac}^0 \ell_{bd}^0} + \frac{ \ell_{bc}^0\ell_{ad}^0      }{\ell_{ac}^0 \ell_{bd}^0}  =1, \quad t>0, 
\end{equation*}
By the converse of Ptolemy's theorem, we can conclude that four points $x_a,x_b,x_c$ and $x_d$ still lie on the same circle (or plane). 
\end{proof}

We remark that Corollary \ref{C5.2} is actually a special case of a more general theorem. See Proposition \ref{P5.3.1}.

Finally, our third application of Theorem \ref{T5.1} is to rule out the existence of periodic solutions to \eqref{Z-1}--\eqref{Z-1-0}. For this, we study the time-evolution of the squared distance functional: for any solution $\{ x_i\}$ to \eqref{Z-1}, we set
\[
\mathcal D_M(X) := \sum_{i,j=1}^N \|x_i-x_j\|^2.
\]
\begin{lemma} \label{L5.1}
Let $\{ x_i(t) \}$ be a solution to the Lohe sphere model \eqref{Z-1}--\eqref{Z-1-0}. Then, $\mathcal D_M(X)$ is non-increasing for $\kappa>0$, and is non-decreasing for $\kappa <0$.
\end{lemma}
\begin{proof}
By direct calculation, one has
\begin{align}
\begin{aligned} \label{Z-4}
&\frac{1}{2}\frac{d}{dt} \|x_i-x_j\|^2 = \langle x_i-x_j, \dot x_i-\dot x_j\rangle \\
& \hspace{0.5cm} = \frac{\kappa}{N}\sum_{k=1}^N  \Big \langle x_i-x_j, \langle x_i,x_i\rangle x_k -\langle x_i,x_k\rangle x_i - \langle x_j,x_j\rangle x_k + \langle x_j,x_k\rangle x_j \Big \rangle  \\
& \hspace{0.5cm} =\frac{\kappa}{N} \Big \langle x_i-x_j, \langle x_i,x_i\rangle N x_c - \langle x_i, N x_c \rangle x_i - \langle x_j,x_j\rangle N x_c + \langle x_j, N x_c \rangle x_j \Big \rangle \\
& \hspace{0.5cm} =\kappa \Big \langle x_i-x_j, \langle x_i,x_i\rangle x_c -\langle x_i, x_c \rangle x_i - \langle x_j,x_j\rangle x_c + \langle x_j, x_c \rangle x_j  \Big \rangle \\
& \hspace{0.5cm} =-\kappa \left[\langle x_j, x_c \rangle + \langle x_i, x_c \rangle - \langle x_j, x_c \rangle\langle x_i,x_j\rangle - \langle x_i,x_c \rangle\langle x_i,x_j\rangle\right].
\end{aligned}
\end{align}
We sum up the relation \eqref{Z-4} with respect to $i,j$ to obtain 
\begin{align}
\begin{aligned} \label{Z-5}
\frac{1}{2} \mathcal D_M(X) &=-\kappa \sum_{i,j=1}^N \Big(  \langle x_j, x_c \rangle + \langle x_i, x_c \rangle - \langle x_j, x_c \rangle\langle x_i,x_j\rangle - \langle x_i, x_c \rangle\langle x_i,x_j\rangle \Big) \\
&=-2N\kappa \left(N\langle x_c, x_c \rangle - \sum_{i=1}^N \langle x_i, x_c \rangle^2 \right). 
\end{aligned}
\end{align}
On the other hand, it follows from the Cauchy-Schwarz inequality that 
\begin{equation} \label{Z-6}
\langle x_i, x_c \rangle^2 \leq \langle x_i,x_i\rangle \langle x_c, x_c \rangle = \langle x_c, x_c \rangle.
\end{equation}
Finally, we combine the relations \eqref{Z-5} and \eqref{Z-6} to derive the desired estimate.
\end{proof}
\begin{remark} \label{R5.1}
\textup{(i)}~~The equality condition of the Cauchy-Schwarz inequality yields that the inequality of \eqref{Z-6} holds if and only if 
\begin{equation*}
|x_c|=0 \quad \textup{or} \quad \textup{all $x_i$ are parallel to $x_c$.}
\end{equation*}
Hence, it follows from   Proposition \ref{P3.1} that
\begin{align*}
\begin{aligned}
\frac{d}{dt} \mathcal D_M(X(t)) = 0 &\quad \Longleftrightarrow \quad \textup{$X(t)$ is an equilibrium solution} \\
& \quad \Longleftrightarrow \quad \frac{d}{dt} X(t)=0.
\end{aligned}
\end{align*}
\textup{(ii)}~~It follows from the unit modulus property of $x_i$ that the functional $D_M(X)$ can be rewritten in terms of the order parameter:
\begin{equation*}
D_M(X) = \sum_{i,j=1}^N \|x_i - x_j\|^2 = 2N^2 (1-\|x_c\|^2).
\end{equation*}
Moreover in \cite{C-H5}, the dynamics of $\|x_c\|$ was derived:
\begin{equation*}
\frac{d}{dt} \|x_c\|^2 = 2\kp \left( \|x_c\|^2 - \frac1N\sum_{i=1}^N \langle x_i,x_c\rangle \langle x_i,x_c\rangle\right),
\end{equation*}
which also yields the same conclusion as Lemma \ref{L5.1}. 
\end{remark}
Finally, we use Lemma \ref{L5.1} and Remark \ref{R5.1} to derive the non-existence of periodic solutions. 
\begin{corollary} \label{C5.3}
System \eqref{Z-1}--\eqref{Z-1-0} does not admit a periodic solution with positive period. 
\end{corollary}
\begin{proof}
Suppose that there exists a periodic solution $X_p$ with positive period $T>0$:
\[  X_p(t+T)= X_p(t) \quad \textup{after some time $t \geq t_0$}. \]
 This implies
\begin{equation} \label{Z-7}
\mathcal D_M(X_p(t+T)) = \mathcal D_M(X_p(t)), \quad t\geq t_0.
\end{equation}
Then, we use \eqref{Z-7} and Lemma \ref{L5.1} to see
\begin{equation*}
\frac{d}{dt} \mathcal D_M(X_p(t)) = 0, \quad t\geq t_0, \quad \Longleftrightarrow \quad \mathcal D_M(X_p(t)) = \mathcal D_M(X_p(t_0)), \quad t\geq t_0.
\end{equation*}
It follows from Remark \ref{R5.1} that $X_p$ is an equilibrium solution to \eqref{Z-1}. In other words, it must stop after $t\geq t_0$. However this implies that $T=0$, and it contradicts the positivity of the period $T$.
\end{proof}
\begin{remark} \label{R5.2}
The proof of Theorems \ref{T3.1}, \ref{T3.2} and \ref{T3.3} do not depend on the sign of the coupling strength $\kappa$. Hence the constants of motion obtained in these three theorem are still valid in the case $\kappa<0$.
\end{remark}
\subsection{Reduction to low-dimensional dynamics} \label{sec:5.2}
In this subsection, we study the low-dimensional dynamics of the Cauchy problem for the Lohe sphere model with frustration:
\begin{equation}\label{SphereLohe}
\begin{cases}
\displaystyle \dot{x}_j=\frac{\kappa}{N}\sum_{k=1}^{N} (Vx_k-\langle x_j,Vx_k\rangle x_j), \quad t > 0, \quad j = 1, \cdots, N, \\
\displaystyle x_j(0)=x_j^0\in \bbs^d,
\end{cases}
\end{equation}
where we permute the indices so that
\begin{equation}\label{NonSingular}
x_j^0\neq x_N^0,~ j =1,\cdots,N-m,\qquad  x_j^0=x_N^0,~j=N-m+1,\cdots,N
\end{equation}
for some $m=1,\cdots,N$. \newline

Next, we seek to generalize the methods of reduction of degree applied to the Kuramoto model in Section \ref{sec:3.2}. Based on our experience with the Kuramoto model, we take the $N$-th variable $x_N$ as our point of reference and employ a stereographic projection about $x_N$, that is, $x_j$ is projected onto the plane $\bbp_{x_N}^{\perp}\le \bbr^{d+1}$ orthogonal to $x_N$, as in the form
\begin{align}
\begin{aligned}\label{Z-8}
y_j &:=x_N+\frac{2}{\|x_j-x_N\|^2}(x_j-x_N)=x_N+\frac{1}{1-\langle x_j,x_N\rangle}(x_j-x_N)\\
&=\frac{1}{1-\langle x_j,x_N\rangle}x_j-\frac{\langle x_j,x_N\rangle}{1-\langle x_j,x_N\rangle}x_N,\quad j=1,\cdots,N-m.
\end{aligned}
\end{align}
For the Kuramoto model, which is the case $d=1$  for the Lohe sphere model, there is a natural group structure on $\bbs^1$ which allows us to apply an orthogonal transformation to normalize $x_N$ and $\bbp_{x_N}^{\perp}$ to a fixed point and a fixed plane, respectively. However, in the case $d\neq 1,3$, there is no such group structure, and thus there is no way to consistently identify the tangent planes. Therefore, we are forced to take the dynamics of $x_N$, and consequently also that of $\bbp_{x_N}^{\perp}$, into account. The inverse transformation is given by
\begin{align}
\begin{aligned}\label{Z-9}
x_j &=x_N+\frac{2}{\|y_j-x_N\|^2}(y_j -x_N)=x_N+\frac{2}{1+\|y_j \|^2}(y_j -x_N)\\
&=\frac{2}{1+\|y_j \|^2}y_j+\frac{-1+\|y_j\|^2}{1+\|y_j \|^2}x_N,\quad j=1,\cdots,N-m.
\end{aligned}
\end{align}
We use the explicit formulae \eqref{Z-8} and \eqref{Z-9} to obtain
\begin{eqnarray*} \label{Z-10}
\langle x_i,x_N\rangle &=& \frac{-1+\|y_i\|^2}{1+\|y_i\|^2},~ 1-\langle x_i,x_N\rangle=\frac{2}{1+\|y_i\|^2},\quad i=1,\cdots,N-m, \\
\langle x_i,Vx_N\rangle &=& \frac{2}{1+\|y_i\|^2}y_i+\frac{-1+\|y_i\|^2}{1+\|y_i\|^2}x_N,~ 1-\langle x_i,x_N\rangle=\frac{2}{1+\|y_i\|^2},\quad i=1,\cdots,N-m, \\
\langle x_i,x_j\rangle &=& \frac{4}{(1+\|y_i\|^2)(1+\|y_j\|^2)}\langle y_i,y_j\rangle+\frac{(-1+\|y_i\|^2)(-1+\|y_j\|^2)}{(1+\|y_i\|^2)(1+\|y_j\|^2)},\quad i \not = j=1,\cdots,N-m, \\
\langle y_i,x_N\rangle &=&0,\quad i=1,\cdots,N-m, \qquad 
\langle x_i,y_j\rangle=\frac{2}{1+\|y_i\|^2}\langle y_i,y_j\rangle,\quad i,j=1,\cdots,N-m.
\end{eqnarray*}

\vspace{0.2cm}

\noindent $\bullet$~(Derivation of the dynamics for $x_i - x_N$): It follows from $\eqref{SphereLohe}_1$ that for $i=1,\cdots,N-1$,
\begin{align}
\begin{aligned}\label{Z-11}
\dot{x}_i=&\frac{\kappa}{N}\sum_{j=1}^{N-m}\Bigg(\frac{2}{1+\|y_j\|^2}Vy_j+\frac{-1+\|y_j\|^2}{1+\|y_j\|^2}Vx_N-\biggl\{ \frac{4}{(1+\|y_i\|^2)(1+\|y_j\|^2)}\langle y_i,y_j\rangle\\
&+\frac{(-1+\|y_i\|^2)(-1+\|y_j\|^2)}{(1+\|y_i\|^2)(1+\|y_j\|^2)}\biggl\}\biggl(\frac{2}{1+\|y_i\|^2}y_i+\frac{-1+\|y_i\|^2}{1+\|y_i\|^2}x_N\biggl)\Bigg)\\
&+\frac{\kappa}{N}\Bigg[x_N-\frac{-1+\|y_i\|^2}{1+\|y_i\|^2}\biggl(\frac{2}{1+\|y_i\|^2}y_i+\frac{-1+\|y_i\|^2}{1+\|y_i\|^2}x_N\biggl)\Bigg]
\end{aligned}
\end{align}
 and
\begin{equation} \label{Z-12}
\dot{x}_N =\frac{\kappa}{N}\sum_{j=1}^{N-1}\Bigg[\left(\frac{2}{1+\|y_j\|^2}y_j+\frac{-1+\|y_j\|^2}{1+\|y_j\|^2}x_N\right)-\frac{-1+\|y_j\|^2}{1+\|y_j\|^2}x_N\Bigg] =\frac{\kappa}{N}\sum_{j=1}^{N-1}\frac{2}{1+\|y_j\|^2}y_j.
\end{equation}
Then, we subtract \eqref{Z-12} from  \eqref{Z-11} to obtain
\begin{align}
\begin{aligned}\label{Z-13}
&\frac{d}{dt}(x_i-x_N) =\frac{\kappa}{N}\sum_{j=1}^{N-1}\Bigg(\frac{-1+\|y_j\|^2}{1+\|y_j\|^2}x_N-\biggl\{ \frac{4}{(1+\|y_i\|^2)(1+\|y_j\|^2)}\langle y_i,y_j\rangle\\
&\hspace{1.5cm} +\frac{(-1+\|y_i\|^2)(-1+\|y_j\|^2)}{(1+\|y_i\|^2)(1+\|y_j\|^2)}\biggl\}\biggl(\frac{2}{1+\|y_i\|^2}y_i+\frac{-1+\|y_i\|^2}{1+\|y_i\|^2}x_N\biggl)\Bigg)\\
&\hspace{0.2cm}+\frac{\kappa}{N}\Bigg[x_N-\frac{-1+\|y_i\|^2}{1+\|y_i\|^2}\biggl(\frac{2}{1+\|y_i\|^2}y_i+\frac{-1+\|y_i\|^2}{1+\|y_i\|^2}x_N\biggl)\Bigg]\\
&\hspace{0.2cm} =\frac{\kappa}{N}\Bigg[-\frac{2(-1+\|y_i\|^2)}{(1+\|y_i\|^2)^2}-\sum_{j=1}^{N-1}\biggl\{ \frac{8\langle y_i,y_j\rangle}{(1+\|y_i\|^2)^2(1+\|y_j\|^2)}+\frac{2(-1+\|y_i\|^2)(-1+\|y_j\|^2)}{(1+\|y_i\|^2)^2(1+\|y_j\|^2)}\biggl\}\Bigg]y_i\\
&\hspace{0.2cm}+\frac{\kappa}{N}\Bigg[\frac{4\|y_i\|^2}{(1+\|y_i\|^2)^2}+\sum_{j=1}^{N-1}\biggl\{ -\frac{4(-1+\|y_i\|^2)\langle y_i,y_j\rangle}{(1+\|y_i\|^2)^2(1+\|y_j\|^2)}+\frac{4\|y_i\|^2(-1+\|y_j\|^2)}{(1+\|y_i\|^2)^2(1+\|y_j\|^2)}\biggl\}\Bigg]x_N.
\end{aligned}
\end{align}

\vspace{0.2cm}

\noindent $\bullet$~(Derivation of the dynamics for $\langle x_i, x_N \rangle$): We use the relation \eqref{C-5}  to get 
\begin{align}
\begin{aligned} \label{Z-14}
&\langle \dot{x}_i,x_N\rangle =\frac{\kappa}{N}\sum_{j=1}^{N-1}\Bigg(\frac{-1+\|y_j\|^2}{1+\|y_j\|^2} \\
& \hspace{0.5cm} -\biggl\{ \frac{4}{(1+\|y_i\|^2)(1+\|y_j\|^2)}\langle y_i,y_j\rangle +\frac{(-1+\|y_i\|^2)(-1+\|y_j\|^2)}{(1+\|y_i\|^2)(1+\|y_j\|^2)}\biggl\}\biggl(\frac{-1+\|y_i\|^2}{1+\|y_i\|^2}\biggl)\Bigg)\\
& \hspace{0.5cm} +\frac{\kappa}{N}\Bigg[1-\frac{-1+\|y_i\|^2}{1+\|y_i\|^2}\biggl(\frac{-1+\|y_i\|^2}{1+\|y_i\|^2}\biggl)\Bigg]\\
& \hspace{0.5cm} =\frac{\kappa}{N}\sum_{j=1}^{N-1}\Bigg(\frac{4\|y_i\|^2(-1+\|y_j\|^2)}{(1+\|y_j\|^2)(1+\|y_i\|^2)^2}- \frac{4\langle y_i,y_j\rangle(-1+\|y_i\|^2)}{(1+\|y_i\|^2)^2(1+\|y_j\|^2)}\Bigg) +\frac{\kappa}{N}\frac{4\|y_i\|^2}{(1+\|y_i\|^2)^2}
\end{aligned}
\end{align}
and
\begin{equation} \label{Z-15}
\langle x_i,\dot{x}_N\rangle = \frac{\kappa}{N}\sum_{j=1}^{N-1}\frac{2}{1+\|y_j\|^2}\langle x_i,y_j\rangle =\frac{\kappa}{N}\sum_{j=1}^{N-1}\frac{4\langle y_i,y_j\rangle}{(1+\|y_i\|^2)(1+\|y_j\|^2)},
\end{equation}
for $i=1,\cdots,N-1$. We use \eqref{Z-14} and \eqref{Z-15} to derive
\begin{align}
\begin{aligned} \label{Z-16}
\frac{d}{dt}\langle x_i,x_N\rangle  &=\langle \dot{x}_i,x_N\rangle+\langle x_i,\dot{x}_N\rangle\\
&=\frac{\kappa}{N}\Bigg[\sum_{j=1}^{N-1}\Bigg(\frac{4\|y_i\|^2(-1+\|y_j\|^2)}{(1+\|y_j\|^2)(1+\|y_i\|^2)^2}+ \frac{8\langle y_i,y_j\rangle}{(1+\|y_i\|^2)^2(1+\|y_j\|^2)}\Bigg)+\frac{4\|y_i\|^2}{(1+\|y_i\|^2)^2}\Bigg].
\end{aligned}
\end{align}
Thus, we combine \eqref{Z-13} and \eqref{Z-16} to obtain 
\begin{align*}
\begin{aligned} \label{Z-17}
&\frac{d}{dt}(x_i-x_N)+(y_i-x_N)\frac{d}{dt}\langle x_i,x_N\rangle\\
&=\frac{\kappa}{N}\Bigg[-\frac{2(-1+\|y_i\|^2)}{(1+\|y_i\|^2)^2}-\sum_{j=1}^{N-1}\biggl\{ \frac{8\langle y_i,y_j\rangle}{(1+\|y_i\|^2)^2(1+\|y_j\|^2)}+\frac{2(-1+\|y_i\|^2)(-1+\|y_j\|^2)}{(1+\|y_i\|^2)^2(1+\|y_j\|^2)}\biggl\}\Bigg]y_i\\
&\quad+\frac{\kappa}{N}\Bigg[\frac{4\|y_i\|^2}{(1+\|y_i\|^2)^2}+\sum_{j=1}^{N-1}\biggl\{ -\frac{4(-1+\|y_i\|^2)\langle y_i,y_j\rangle}{(1+\|y_i\|^2)^2(1+\|y_j\|^2)}+\frac{4\|y_i\|^2(-1+\|y_j\|^2)}{(1+\|y_i\|^2)^2(1+\|y_j\|^2)}\biggl\}\Bigg]x_N\\
&\quad+\frac{\kappa}{N}\Bigg[\sum_{j=1}^{N-1}\Bigg(\frac{4\|y_i\|^2(-1+\|y_j\|^2)}{(1+\|y_j\|^2)(1+\|y_i\|^2)^2}+ \frac{8\langle y_i,y_j\rangle}{(1+\|y_i\|^2)^2(1+\|y_j\|^2)}\Bigg)+\frac{4\|y_i\|^2}{(1+\|y_i\|^2)^2}\Bigg](y_i-x_N)\\
&=\frac{\kappa}{N}\Bigg[\frac{2}{1+\|y_i\|^2}+\sum_{j=1}^{N-1}\frac{2(-1+\|y_j\|^2)}{(1+\|y_i\|^2)(1+\|y_j\|^2)}\Bigg]y_i +\frac{\kappa}{N}\Bigg[-\sum_{j=1}^{N-1} \frac{4\langle y_i,y_j\rangle}{(1+\|y_i\|^2)(1+\|y_j\|^2)}\Bigg]x_N,
\end{aligned}
\end{align*}
where $i = 1, \cdots, N-1$. 
We summarize the discussion abovet in the following proposition.
\begin{lemma} \label{L5.2}
The Cauchy problem \eqref{SphereLohe}--\eqref{NonSingular} is equivalent to the following Cauchy Problem:
\begin{equation}\label{StereoLohe}
\begin{cases}
\displaystyle \dot{y}_i=\frac{\kappa}{N}\sum_{j=1}^{N-1}\frac{2}{1+\|y_j\|^2}y_j+\frac{\kappa}{N}\Bigg[1+\sum_{j=1}^{N-1}\frac{(-1+\|y_j\|^2)}{1+\|y_j\|^2}\Bigg]y_i +\frac{\kappa}{N}\Bigg[-\sum_{j=1}^{N-1} \frac{2\langle y_i,y_j\rangle}{1+\|y_j\|^2}\Bigg]x_N,  \\
\displaystyle \dot{x}_N=\frac{\kappa}{N}\sum_{j=1}^{N-1}\frac{2}{1+\|y_j\|^2}y_j, \quad  t > 0,\\
y_i(0)=x_N^0+\frac{2}{\|x_i^0-x_N^0\|^2}(x_i^0-x_N^0)\in \bbp_{x_N^0}^{\perp},\quad i=1.\cdots,N-1,\\
x^N(0)=x_N^0\in \bbs^d,
\end{cases}
\end{equation}
subject to constraints:
\begin{equation*}\label{nonsingular2}
\begin{cases}
y_i(0)\neq y_j(0),\quad i,j=1,\cdots,N-1,~i\neq j,\\
y_i(0)\neq \infty, \quad i=1,\cdots,N-1.
\end{cases}
\end{equation*}
\end{lemma}
\begin{proof} We use \eqref{Z-13} and \eqref{Z-16} to get 
\begin{align*}
\begin{aligned}
\dot{y}_i &=\dot{x}_N+\frac{1}{1-\langle x_i,x_N\rangle}\frac{d}{dt}(x_i-x_N)+\frac{1}{(1-\langle x_i,x_N\rangle)^2}(x_i-x_N)\frac{d}{dt}\langle x_i,x_N\rangle\\
&=\dot{x}_N+\frac{1+\|y_i\|^2}{2}\frac{d}{dt}(x_i-x_N)+\frac{(1+\|y_i\|^2)^2}{4}\cdot \frac{2}{1+\|y_i\|^2}(y_i-x_N)\frac{d}{dt}\langle x_i,x_N\rangle\\
&=\dot{x}_N+\frac{1+\|y_i\|^2}{2}\Bigg[\frac{d}{dt}(x_i-x_N)+(y_i-x_N)\frac{d}{dt}\langle x_i,x_N\rangle\Bigg]\\
&=\frac{\kappa}{N}\sum_{j=1}^{N-1}\frac{2}{1+\|y_j\|^2}y_j +\frac{1+\|y_i\|^2}{2}\cdot \frac{\kappa}{N}\Bigg[\frac{2}{1+\|y_i\|^2}+\sum_{j=1}^{N-1}\frac{2(-1+\|y_j\|^2)}{(1+\|y_i\|^2)(1+\|y_j\|^2)}\Bigg]y_i\\
&\hspace{0.5cm}+\frac{1+\|y_i\|^2}{2}\cdot\frac{\kappa}{N}\Bigg[-\sum_{j=1}^{N-1} \frac{4\langle y_i,y_j\rangle}{(1+\|y_i\|^2)(1+\|y_j\|^2)}\Bigg]x_N\\
&=\frac{\kappa}{N}\sum_{j=1}^{N-1}\frac{2}{1+\|y_j\|^2}y_j+\frac{\kappa}{N}\Bigg[1+\sum_{j=1}^{N-1}\frac{(-1+\|y_j\|^2)}{1+\|y_j\|^2}\Bigg]y_i+\frac{\kappa}{N}\Bigg[-\sum_{j=1}^{N-1} \frac{2\langle y_i,y_j\rangle}{1+\|y_j\|^2}\Bigg]x_N.
\end{aligned}
\end{align*}
\end{proof}
\begin{lemma}\label{L5.3}
Let $\{y_j \}$ be a solution to \eqref{StereoLohe}. For any four indices $i, j, k, l=1,\cdots,N-1$, we have
\[
\frac{d}{dt}\langle y_i-y_j,y_k-y_l\rangle=\frac{2\kappa}{N}\Bigg[1+\sum_{m=1}^{N-1}\frac{-1+\|y_m\|^2}{1+\|y_m\|^2}\Bigg]\langle y_i-y_j,y_k-y_l\rangle.
\]
\end{lemma}
\begin{proof} Note that 
\begin{align}
\begin{aligned} \label{Z-18}
& \dot{y}_i-\dot{y}_j =\frac{\kappa}{N}\Bigg[1+\sum_{m=1}^{N-1}\frac{-1+\|y_m\|^2}{1+\|y_m\|^2}\Bigg](y_i-y_j)-\frac{\kappa}{N}\Bigg[\sum_{m=1}^{N-1} \frac{2\langle y_i-y_j,y_m\rangle}{1+\|y_m\|^2}\Bigg]x_N, \\
& \dot{y}_k-\dot{y}_l =\frac{\kappa}{N}\Bigg[1+\sum_{m=1}^{N-1}\frac{-1+\|y_m\|^2}{1+\|y_m\|^2}\Bigg](y_k-y_l)-\frac{\kappa}{N}\Bigg[\sum_{m=1}^{N-1} \frac{2\langle y_k-y_l,y_m\rangle}{1+\|y_m\|^2}\Bigg]x_N.
\end{aligned}
\end{align}
Then, we use \eqref{Z-18} to find
\begin{align*}
\begin{aligned}
& \frac{d}{dt}\langle y_i-y_j,y_k-y_l\rangle  = \langle \dot{y}_i-\dot{y}_j,y_k-y_l\rangle+\langle y_i-y_j,\dot{y}_k-\dot{y}_l\rangle\\
& \hspace{0.5cm}  =\frac{\kappa}{N}\Bigg[1+\sum_{m=1}^{N-1}\frac{-1+\|y_m\|^2}{1+\|y_m\|^2}\Bigg]\langle y_i-y_j,y_k-y_l\rangle +\frac{\kappa}{N}\Bigg[1+\sum_{m=1}^{N-1}\frac{-1+\|y_m\|^2}{1+\|y_m\|^2}\Bigg]\langle y_i-y_j,y_k-y_l\rangle \\
& \hspace{0.5cm}  =\frac{2\kappa}{N}\Bigg[1+\sum_{m=1}^{N-1}\frac{-1+\|y_m\|^2}{1+\|y_m\|^2}\Bigg]\langle y_i-y_j,y_k-y_l\rangle.
\end{aligned}
\end{align*}
\end{proof}
\begin{proposition}
For any eight indices $a, b, c, d, e, f, g, h=1,\cdots, N-1$, we have
\[
\langle y_a-y_b,y_c-y_d\rangle (t) \langle y_e-y_f,y_g-y_h\rangle (0) = \langle y_a-y_b,y_c-y_d\rangle(0) \langle y_e-y_f,y_g-y_h\rangle(t)
\]
for all times $t\ge 0$.
\end{proposition}
\begin{proof} We use Lemma \ref{L5.3} to get
\begin{align*}
\begin{aligned}
& \langle y_a-y_b,y_c-y_d\rangle (t)=\langle y_a-y_b,y_c-y_d\rangle (0)\exp\Bigg[\frac{2\kappa}{N}\int_0^t \Bigg(1+\sum_{k=1}^{N-1}\frac{-1+\|y_k\|^2}{1+\|y_k\|^2}\Bigg)(s)ds\Bigg], \\
& \langle y_e-y_f,y_g-y_h\rangle (t)=\langle y_e-y_f,y_g-y_h\rangle (0)\exp\Bigg[\frac{2\kappa}{N}\int_0^t \Bigg(1+\sum_{k=1}^{N-1}\frac{-1+\|y_k\|^2}{1+\|y_k\|^2}\Bigg)(s)ds\Bigg].
\end{aligned}
\end{align*}
These yield the desired estimate.

\end{proof}

\vspace{0.5cm}

Now, we are ready to provide the low-dimensional dynamics for $y_i$. For each time $t\ge 0$, we may select three quantities $M(t)\in O(d+1)$, $a(t)>0$, $b(t)\in \bbp_{x_N^0}^\perp$ which satisfy
\begin{equation}\label{data}
\begin{cases}
y_i(t)=M(t)(a(t)y_i(0)+b(t)),\quad i=1,\cdots,N-1,\\
x_N(t)=M(t)x_N^0, \\
M(0)=I_{d+1},\quad a(0)=1,\quad b(0)=0\in\bbp_{x_N^0}^\perp.
\end{cases}
\end{equation}
Note that the three quantities $M(t)$, $a(t)$ and $b(t)$ along with the initial data $y_1(0),\cdots,y_{N-1}(0),x_N^0$ fully describe the behavior of $y_1(t),\cdots,y_{N-1},x_N(t)$. In what follows, we will heuristically derive the dynamics of $M(t)$, $a(t)$, $b(t)$, and as a consequence, we  provide a unique choice for $M(t)$ and $b(t)$, under the assumption that
\begin{equation*}\label{span}
\{ y_1(0)-y_{N-1}(0),\cdots,y_{N-2}-y_{N-1}(0)\} ~~ \textrm{ spans } ~~ \bbp_{x_N^0}^\perp.
\end{equation*}
(These assumptions will eventually be unnecessary.) Of course, this requires $N\ge d+2$, since $\bbp_{x_N^0}^\perp$ is $d$-dimensional. Note that 
\begin{align*}
\begin{aligned} \label{Z-19-1}
\langle y_a-y_b,y_c-y_d\rangle (t)=&\langle M(t)a(t)(y_a(0)-y_b(0)),M(t)a(t)(y_c(0)-y_d(0))\rangle\\
=&a(t)^2\langle y_a(0)-y_b(0),y_c(0)-y_d(0)\rangle,
\end{aligned}
\end{align*}
and consequently
\begin{equation}\label{Z-19-2}
\frac{d}{dt}\langle y_a-y_b,y_c-y_d\rangle (t) = 2a(t)a'(t)\langle y_a(0)-y_b(0),y_c(0)-y_d(0)\rangle =\frac{2a'(t)}{a(t)}\langle y_a-y_b,y_c-y_d\rangle (t).
\end{equation}
Now we compare \eqref{Z-19-2} with Lemma \ref{L5.3} to get 
\begin{equation*}
\frac{d}{dt}\langle y_a-y_b,y_c-y_d\rangle=\frac{2\kappa}{N}\Bigg[1+\sum_{k=1}^{N-1}\frac{-1+\|y_k\|^2}{1+\|y_k\|^2}\Bigg]\langle y_a-y_b,y_c-y_d\rangle.
\end{equation*}
This yields
\begin{equation*}
\frac{2a'(t)}{a(t)}=\frac{2\kappa}{N}\Bigg[1+\sum_{k=1}^{N-1}\frac{-1+\|y_k\|^2}{1+\|y_k\|^2}\Bigg],
\end{equation*}
or equivalently
\begin{equation}\label{DEa}
a'(t)=\frac{\kappa}{N}\Bigg[1+\sum_{k=1}^{N-1}\frac{-1+\|y_k\|^2}{1+\|y_k\|^2}\Bigg]a(t).
\end{equation}
Next, we substitute \eqref{data} into \eqref{StereoLohe} to obtain
\begin{align}
\begin{aligned} \label{Z-19-3}
&M'(t)(a(t)y_i(0)+b(t))+M(t)(a'(t)y_i(0)+b'(t))\\
& =\frac{\kappa}{N}\sum_{k=1}^{N-1}\frac{2}{1+\|y_k\|^2}M(t)(a(t)y_k(0)+b(t)) +\frac{\kappa}{N}\Bigg[1+\sum_{k=1}^{N-1}\frac{-1+\|y_k\|^2}{1+\|y_k\|^2}\Bigg]M(t)(a(t)y_i(0)+b(t))\\
& \hspace{0.2cm}+\frac{\kappa}{N}\Bigg[-\sum_{k=1}^{N-1} \frac{2\langle y_i,y_k\rangle}{1+\|y_k\|^2}\Bigg]M(t)x_N^0,\quad i=1,\cdots,N-1,
\end{aligned}
\end{align}
and
\begin{equation} \label{Z-19-4}
M'(t)x_N^0=\frac{\kappa}{N}\sum_{k=1}^{N-1}\frac{2}{1+\|y_k\|^2}M(t)(a(t)y_k(0)+b(t)).
\end{equation}
We multiply $M(t)^{-1}$ to the left sides of \eqref{Z-19-3} and \eqref{Z-19-4} to get
\begin{align}
\begin{aligned}\label{DE1}
&M(t)^{-1}M'(t)(a(t)y_i(0)+b(t))+a'(t)y_i(0)+b'(t)\\
& \hspace{0.5cm} =\frac{\kappa}{N}\sum_{k=1}^{N-1}\frac{2}{1+\|y_k\|^2}(a(t)y_k(0)+b(t))+\frac{\kappa}{N}\Bigg[1+\sum_{k=1}^{N-1}\frac{-1+\|y_k\|^2}{1+\|y_k\|^2}\Bigg](a(t)y_i(0)+b(t))\\
&\hspace{0.7cm} +\frac{\kappa}{N}\Bigg[-\sum_{k=1}^{N-1} \frac{2\langle y_i,y_k\rangle}{1+\|y_k\|^2}\Bigg]x_N^0,\quad i=1,\cdots,N-1,
\end{aligned}
\end{align}
and
\begin{equation}\label{DE2}
M(t)^{-1}M'(t)x_N^0=\frac{\kappa}{N}\sum_{j=1}^{N-1}\frac{2}{1+\|y_k\|^2}(a(t)y_k(0)+b(t)).
\end{equation}
We now take the difference of \eqref{DE1} for $i=1,\cdots,N-2$ and \eqref{DE1} for $N-1$ to obtain that for $ i=1,\cdots,N-1$, 
\begin{align*}
\begin{aligned}
&M(t)^{-1}M'(t)a(t)(y_i(0)-y_{N-1}(0))+a'(t)(y_i(0)-y_{N-1}(0))\\
& \hspace{0.5cm} =\frac{\kappa}{N}\Bigg[1+\sum_{k=1}^{N-1}\frac{-1+\|y_k\|^2}{1+\|y_k\|^2}\Bigg]a(t)(y_i(0)-y_{N-1}(0)) -\frac{\kappa}{N}\Bigg[\sum_{k=1}^{N-1} \frac{2\langle y_i-y_{N-1},y_k\rangle}{1+\|y_k\|^2}\Bigg]x_N^0.
\end{aligned}
\end{align*}
We substitute $a'(t)$ for \eqref{DEa} to see that for $i = 1, \cdots, N-1$, 
\begin{align*}
\begin{aligned} 
&M(t)^{-1}M'(t)a(t)(y_i(0)-y_{N-1}(0)) =-\frac{\kappa}{N}\Bigg[\sum_{k=1}^{N-1} \frac{2\langle y_i-y_{N-1},y_k\rangle}{1+\|y_k\|^2}\Bigg]x_N^0 \\
& \hspace{0.5cm}  =-\frac{\kappa}{N}\Bigg[\sum_{k=1}^{N-1} \frac{2a(t) \langle y_i(0)-y_{N-1}(0),a(t)y_k(0)+b(t)\rangle}{1+\|y_k\|^2}\Bigg]x_N^0, \\
& \hspace{0.5cm} \textup{and} \quad \textup{span} \{ y_1(0)-y_{N-1}(0),\cdots,y_{N-2}-y_{N-1}(0) \} = \bbp_{x_N^0}^\perp.
\end{aligned}
\end{align*}
Thus, one has
\begin{align}
\begin{aligned} \label{Z-19-6}
& M(t)^{-1}M'(t)a(t)v =-\frac{\kappa}{N}\Bigg[\sum_{k=1}^{N-1} \frac{2a(t) \langle v,a(t)y_k(0)+b(t)\rangle}{1+\|y_k\|^2}\Bigg]x_N^0,\quad v\in \bbp_{x_N^0}^\perp, \\
&  \textup{or equivalently,} \quad  M(t)^{-1}M'(t)v =-\frac{\kappa}{N}\Bigg[\sum_{k=1}^{N-1} \frac{2 \langle v,a(t)y_k(0)+b(t)\rangle}{1+\|y_k\|^2}\Bigg]x_N^0,\quad v\in \bbp_{x_N^0}^\perp.
\end{aligned}
\end{align}
We combine \eqref{Z-19-6} and \eqref{DE2} to get 
\begin{align}
\begin{aligned}\label{DEM}
&M(t)^{-1}M'(t)v \\
& \hspace{0.2cm} =\frac{\kappa}{N}\sum_{j=1}^{N-1}\frac{2}{1+\|y_k\|^2}(a(t)y_k(0)+b(t))\langle v,x_N^0\rangle -\frac{\kappa}{N}\Bigg[\sum_{k=1}^{N-1} \frac{2 \langle v,a(t)y_k(0)+b(t)\rangle}{1+\|y_k\|^2}\Bigg]x_N^0\\
& \hspace{0.2cm}  =\langle v,x_N^0\rangle z- \langle v, z\rangle x_N^0,\quad v\in \bbr^{d+1},
\end{aligned}
\end{align}
where
\[
z(t):=\frac{K}{N}\sum_{j=1}^{N-1}\frac{2}{1+\|y_k\|^2}(a(t)y_k(0)+b(t))\in \bbp_{x_N^0}^\perp.
\]
Note that $M(t)^{-1}M'(t)\in \mathfrak{o}(d+1)$, that is, $M(t)^{-1}M'(t)$ should be anti-symmetric. This will be verified in the following lemma.
\begin{lemma}\label{L5.4}
Let $(V,\langle\cdot,\cdot \rangle)$ be an inner product space over $\bbr$, and $x,z\in V$ vectors such that $\langle x, z\rangle =0$. Then the linear operator $L:V\rightarrow V$ defined by
\begin{equation} \label{Z-20}
L(v)=\langle v,x\rangle z- \langle v, z\rangle x,\quad v\in V
\end{equation}
is anti-symmetric:
\[
\langle v,L(w)\rangle=-\langle L(v),w\rangle,\quad v,w\in V.
\]
\end{lemma}
\begin{proof} We use defining relation \eqref{Z-20} to find 
\begin{align*}
\begin{aligned}
\langle v,L(w)\rangle+\langle L(v),w\rangle &= \langle v,\langle w,x\rangle z-\langle w,z\rangle x\rangle+\langle \langle v,x\rangle z-\langle v,z\rangle x,w\rangle\\
& =\langle w,x\rangle  \langle v,z\rangle-\langle w,z\rangle \langle v,x\rangle +\langle v,x\rangle \langle z, w\rangle-\langle v,z\rangle \langle x,w\rangle =0.
\end{aligned}
\end{align*}
This yields the desired anti-symmetry of $L$. 
\end{proof}

\begin{proposition} \label{P5.2}
Let $a(t)\in \bbr$, $b(t)\in \bbp_{x_N^0}^\perp$ and $M(t)\in O(d+1)$ be solutions to the Cauchy problem:
\begin{equation}\label{reducedSphereLohe}
\begin{cases}
\displaystyle a'(t)=\frac{\kappa}{N}\Bigg[1+\sum_{k=1}^{N-1}\frac{-1+\|a(t)y_k(0)+b(t)\|^2}{1+\|a(t)y_k(0)+b(t)\|^2}\Bigg]a(t),\\
\displaystyle b'(t)= \kappa b(t)+\frac{\kappa}{N}\sum_{k=1}^{N-1}\frac{2}{1+\|a(t)y_k(0)+b(t)\|^2}a(t)y_k(0),\\
\displaystyle M'(t)=M(t)L(t),\\
a(0)=1, \quad b(0)=0\in \bbp_{x_N^0}^\perp, \quad M(0)=I_{d+1},
\end{cases}
\end{equation}
where $L(t)=L(t,a,b,y_1(0),\cdots,y_{N-1}(0),x_N^0)$ and $z$ are the anti-symmetric operator and quanitity defined by the following relations:
\begin{align*}
\begin{aligned}
& L(t)v =\langle v,x_N^0\rangle z(t)- \langle v, z(t) \rangle x_N^0,\quad v\in \bbr^{d+1}, \\
& z(t) :=\frac{\kappa}{N}\sum_{j=1}^{N-1}\frac{2}{1+\|a(t)y_k(0)+b(t)\|^2}(a(t)y_k(0)+b(t))\in \bbp_{x_N^0}^\perp.
\end{aligned}
\end{align*}
Then, the relations \eqref{data} hold.
\end{proposition}
\begin{proof} The dynamics for $a$ and $M$ follow from \eqref{DEa} and \eqref{DEM}. On the other hand, the dynamics of $b$ follows from \eqref{DE1} using \eqref{DEa} and \eqref{DEM} that
\begin{align*}
\begin{aligned} \label{Z-21}
b'(t) &=\frac{\kappa}{N}\sum_{k=1}^{N-1}\frac{2}{1+\|y_k\|^2}(a(t)y_k(0)+b(t))+\frac{\kappa}{N}\Bigg[1+\sum_{k=1}^{N-1}\frac{-1+\|y_k\|^2}{1+\|y_k\|^2}\Bigg](a(t)y_i(0)+b(t))\\
&\hspace{0.5cm}+\frac{\kappa}{N}\Bigg[-\sum_{k=1}^{N-1} \frac{2\langle y_i,y_k\rangle}{1+\|y_k\|^2}\Bigg]x_N^0-M(t)^{-1}M'(t)(a(t)y_i(0)+b(t))-a'(t)y_i(0)\\
&=\frac{\kappa}{N}\sum_{k=1}^{N-1}\frac{2}{1+\|y_k\|^2}(a(t)y_k(0)+b(t))+\frac{\kappa}{N}\Bigg[1+\sum_{k=1}^{N-1}\frac{-1+\|y_k\|^2}{1+\|y_k\|^2}\Bigg]b(t)\\
&\hspace{0.5cm}+\left\{\frac{\kappa}{N}\Bigg[-\sum_{k=1}^{N-1} \frac{2\langle a(t)y_i(0)+b(t),a(t)y_k(0)+b(t)\rangle}{1+\|y_k\|^2}\Bigg]x_N^0-M(t)^{-1}M'(t)(a(t)y_i(0)+b(t))\right\}\\
&\hspace{0.5cm}+\left\{\frac{\kappa}{N}\Bigg[1+\sum_{k=1}^{N-1}\frac{-1+\|y_k\|^2}{1+\|y_k\|^2}\Bigg]a(t)-a'(t)\right\}y_i(0)\\
&=\frac{\kappa}{N}\sum_{k=1}^{N-1}\frac{2}{1+\|y_k\|^2}(a(t)y_k(0)+b(t))+\frac{\kappa}{N}\Bigg[1+\sum_{k=1}^{N-1}\frac{-1+\|y_k\|^2}{1+\|y_k\|^2}\Bigg]b(t)\\
&=\frac{\kappa}{N}\sum_{k=1}^{N-1}\frac{2}{1+\|y_k\|^2}a(t)y_k(0)+ \kappa b(t).
\end{aligned}
\end{align*}
\end{proof}
\begin{remark} We provide two comments below. \newline

\noindent (i)~We have reduced the dynamics of \eqref{SphereLohe} on $(\bbs^d)^N$ to the dynamics of \eqref{reducedSphereLohe} on $(0,\infty)\times \bbp_{x_N^0}^\perp\times SO(d+1)$, a $\frac{(d+1)(d+2)}{2}$-dimensional manifold. In the latter case, the degree of freedom of the system is manifest in the governing equations, not the initial data. \newline

\noindent (ii)~Note that the differential equations for $a$ and $b$, do not depend on $M$. Hence we have a hierarchy: we can first solve for $a$ and $b$, and then solve for $M$ by integration.

Moreover, as far as synchronization properties based on Euclidean distances are concerned, the orthogonal transformation $M(t)$ plays no role. Thus analyzing the emergent dynamics of $(a,b)\in (0,\infty)\times \bbp_{x_N^0}^\perp$ on a $(d+1)$-dimensional manifold is enough to determine whether asymptotic aggregation occurs.
\end{remark}
As an example of how only $a$ and $b$(and not $M$) matters when it comes to aggregation estimates, we consider the order parameter. For a given position configuration $\{x_j \}$, we introduce the order parameter $\rho$:
\[  x_c := \frac{1}{N} \sum_{k=1}^{N} x_k, \qquad  \rho := \| x_c \|. \]
By direct calculation, one has 
\begin{align*}
\begin{aligned} \label{Z-22}
\rho^2=&\|\frac{1}{N}\sum_{k=1}^N x_k\|^2=\left\|\sum_{k=1}^{N-1}\frac{2}{1+\|y_k\|^2}y_k+\Big[1+\sum_{k=1}^{N-1}\frac{-1+\|y_k\|^2}{1+\|y_k\|^2}\Big]x_N\right\|^2\\
=&\left\|\sum_{k=1}^{N-1}\frac{2}{1+\|y_k\|^2}y_k\right\|^2+\left(1+\sum_{k=1}^{N-1}\frac{-1+\|y_k\|^2}{1+\|y_k\|^2}\right)^2\\
=&\left\|\sum_{k=1}^{N-1}\frac{2(a(t)y_k(0)+b(t))}{1+\|a(t)y_k(0)+b(t)\|^2}\right\|^2+\left(1+\sum_{k=1}^{N-1}\frac{-1+\|a(t)y_k(0)+b(t)\|^2}{1+\|a(t)y_k(0)+b(t)\|^2}\right)^2.
\end{aligned}
\end{align*}

In the following proposition, we show that $\rho$ is non-decreasing along the flow \eqref{SphereLohe}.
\begin{proposition}
The order parameter $\rho^2$ is a non-decreasing functional along the flow \eqref{SphereLohe}.
\end{proposition}
\begin{proof}
We first observe that equation \eqref{SphereLohe} is linear with respect to the frustration matrix $V$. Thus, if we set $\tilde V := V/ \|V\|_\textup{op}$, then \eqref{SphereLohe} becomes
\begin{equation*}
\dot x_j = \frac{\kp \|V\|_\textup{op}}{N}\sum_{k=1}^N \Big( \tilde V x_k - \langle x_j, Vx_k\rangle x_j\Big).
\end{equation*}
Hence, without loss of generality, we may assume that the operator norm of a frustration matrix $V$ is 1. On the other hand, direct calculation yields 
\begin{align*}
\frac12\frac{d}{dt} \|x_c\|^2 &= \kp \left( \|x_c\|^2 - \frac1N\sum_{k=1}^N \langle x_k, Vx_c\rangle \langle x_k,x_c\rangle\right)  \geq \kp ( 1- \|V\|_\textup{op} ) \|x_c\|^2 =0.
\end{align*}
This shows the desired nondecreasing property of the order parameter.
\end{proof}

As another application of the constants of motion, we note the invariance of subspaces intersecting $\bbs^d$, or in stereographic coordinates, affine subspaces of $\bbr^{d+1}$. This generalizes Corollary \ref{C5.2}, which corresponds to the case of 2-dimensional affine surfaces. This also shows the advantage of using the viewpoint of stereographic projection over the M\"obius transformation, suggested in the beginning of Section \ref{sec:3.2}.

\begin{proposition}\label{P5.3.1}
Consider the Lohe sphere model \eqref{SphereLohe} with initial data $x_i^0\in \bbs^d$ satisfying condition \eqref{NonSingular}. Suppose $n,m\in \bbn$ so that the $n$ points $x_1^0,\cdots,x^0_n$ lie in an $m$-dimensional affine subspace of $\bbr^{d+1}$. Then for all $t\ge 0$, $x_1(t),\cdots,x_n(t)$ lie in some $m$-dimensional affine subspace of $\bbr^{d+1}$.
\end{proposition}
\begin{proof}
If $n=1$, there is nothing to prove, so assume $n\ge 2$. We may relabel the indices so that the points in consideration are $x_1,\cdots,x_{n-1}$ and $x_N$. Then the condition that $x_1,\cdots,x_{n-1}$ and $x_N$ lie in some $m$-dimensional affine subspace of $\bbr^{d+1}$ is equivalent to the condition that $y_1,\cdots,y_{n-1}$ lie in some $(m-1)$-dimensional affine subspace of $\bbp_{x_N}^\perp$. However, since $y_i(t)=M(t)(a(t)y_i(0)+b(t))$ from \eqref{data}, the foregoing assertion is equivalent to saying that $a(t)y_i(0)+b(t)$, $i=1,\cdots,n-1$ lie in some $(m-1)$-dimensional affine subspace of $\bbp_{x_N^0}^\perp$. Since this is true for $t=0$, this must be true for all $t\ge 0$.
\end{proof}

\subsection{Complete aggregation} \label{sec:5.3} In this subsection, we  present an exponential aggregation  estimate of the Lohe sphere model with frustration. This is be clearly expected from the lower dimensional case of Section \ref{sec:3.1}.

\begin{theorem} \label{T5.2}
Suppose that the natural frequency matrices, frustration matrix and initial data satisfy 
\begin{align}
\begin{aligned} \label{Z-22-1}
& V = aI_d + W, \quad \| W \|<a \ll 1, \quad  \Omega_i \equiv \Omega \quad \textup{for all $i=1,\cdots,N$}, \\
& \max_{1\leq i,j\leq N} \Big(1-\langle x_i^0,x_j^0\rangle \Big) < 1- \frac{ \|W\|}{a},
\end{aligned}
\end{align}
where $W$ is a $d\times d$ skew-symmetric matrix, and let $\{x_i\}_{i=1}^N$ be a solution to \eqref{SphereLohe}. Then, one has complete aggregation:
\[ \lim_{t\to\infty} \max_{1\leq i,j\leq N} \|x_i(t) -x_j(t)\|  = 0. \]
\end{theorem}
\begin{proof} We substitute the relation $V=aI_d +  W$ in  \eqref{SphereLohe} to get 
\begin{equation*} \label{Z-23}
\dot x_i = \frac{a\kp}{N} \sum_{k=1}^N (x_k - \langle x_i,x_k\rangle x_i) + \frac{\kp}{N} \sum_{k=1}^N (   Wx_k -\langle x_i , Wx_k \rangle x_i ).
\end{equation*}
We set 
\[  R_{ij}:= \langle x_i,x_j\rangle, \quad 1 \leq i, j \leq N. \]
Then, one has $|R_{ij}| \leq 1$ and 
\begin{equation} \label{Z-24}
\frac{d}{dt} R_{ij} = \frac{a\kp}{N}\sum_{k=1}^N(R_{ik}+R_{kj})(1-R_{ij}) + \frac{\kp}{N} \sum_{k=1}^N (\langle x_i, Wx_k\rangle + \langle x_j, Wx_k \rangle)(1-R_{ij}).
\end{equation}
Since we expect $R_{ij} \to 1$ as $t \to \infty$, it will be convenient to use $A_{ij}$:
\[ A_{ij} := 1-R_{ij}. \]
It follows from the definition of $A_{ij}$ and  \eqref{Z-24} that 
\begin{align*}
\frac{d}{dt} A_{ij} &= -2a\kp A_{ij} + \frac{a\kp}{N}\sum_{k=1}^N (A_{ik}+A_{kj})A_{ij} + \frac{\kp}{N} \sum_{k=1}^N (\langle x_i, Wx_k\rangle + \langle x_j, Wx_k \rangle)A_{ij} \\
&\leq -2a\kp A_{ij} + \frac{a\kp}{N}\sum_{k=1}^N (A_{ik}+A_{kj})A_{ij} + 2\kp \|  W\| A_{ij}.
\end{align*}
For each $t>0$, we choose the extremal indices $(i_t,j_t)$ satisfying the relation:
\begin{equation*}
D(\mathcal A(t)) := \max_{1\leq i,j\leq N} A_{ij} = A_{i_tj_t}.
\end{equation*}
Then, $D(\mathcal A)$ satisfies 
\[
\frac{d}{dt} D(\mathcal A) \leq -2\kp(a-\| W\| )D(\mathcal A) + 2a\kp D(\mathcal A)^2.
\]
Finally, we use Gr\"onwall's lemma and the smallness condition \eqref{Z-22-1} to derive the desired result.
\end{proof}

\subsubsection{A two-oscillator system} For motivation, consider the two-oscillator system with a special ansatz for frustration matrix:
\[
V = aI_d +  W,
\]
where $W$ is a $d\times d$ skew-symmetric matrix. First, we consider the special case $a=0$. In other words, we assume that the frustration matrix $W$ is given to be skew-symmetric. Then, the two-oscillator system becomes
\begin{align}
\begin{aligned} \label{Y-2}
\dot x_1 &= \frac{\kp}{2} ( Wx_2 - \langle x_1,Wx_2\rangle x_1), \\
\dot x_2 &= \frac{\kp}{2} ( Wx_1 - \langle x_2,Wx_1\rangle x_2).
\end{aligned}
\end{align}
Then, it is easy to see that $R_{12}:=\langle x_1,x_2\rangle$ satisfies
\begin{align*} \label{Y-3}
\begin{aligned}
\frac{d}{dt} R_{12} &= \frac{\kp}{2} \Big( \langle Wx_2,x_2\rangle - \langle x_1,Wx_2\rangle \langle x_1,x_2\rangle + \langle Wx_1,x_1\rangle - \langle x_2,Wx_1\rangle \langle x_1,x_2\rangle\Big) \\
&=\frac{\kp}{2} \Big( \langle Wx_1,x_2\rangle - \langle Wx_1,x_2\rangle \Big) \langle x_1,x_2\rangle = 0,
\end{aligned}
\end{align*}
where we use the fact that $W$ is skew-symmetric:
\begin{equation*}
\langle Wy,y \rangle = 0\quad\textup{and} \quad \langle Wy,z \rangle = -\langle y, Wz\rangle.
\end{equation*}
Hence we have the following proposition.
\begin{proposition} \label{P5.4}
Let $(x_1,x_2)$ be a solution to \eqref{Y-2} with initial data $(x_1^0,x_2^0)$. Then we have
\begin{equation*}
\langle x_1,x_2\rangle(t) = \langle x_1^0,x_2^0\rangle, \quad t \geq 0.
\end{equation*}
\end{proposition}
In other words, the relative angle between the position $x_1$ and $x_2$ is conserved under the flow \eqref{Y-2}.
\begin{remark} \label{RY.1}
Complete synchronization between $x_1$ and $x_2$ is equivalent to $\langle x_1,x_2\rangle$ converging to $1$. However, in the presence of skew-symmetric frustration, Proposition \ref{P5.4} tells us that their inner product is always constant. Hence, unless $x_1$ and $x_2$ are located in the same point initially, then they cannot be close to each other and complete synchronization is \textcolor{black}{impossible}.
\end{remark}

\subsubsection{A many-oscillator system} In this part, we consider the many-oscillator case. For a spatial configuration $\{ x_j \}$, we set 
\[ R_{ij}:=\langle x_i,x_j\rangle, \quad 1 \leq i, j \leq N. \]
\begin{theorem} \label{T5.3}
Let $\{x_i\}_{i=1}^N$ be a solution to the Lohe sphere model \eqref{SphereLohe} with skew-symmetric frustration matrix. Then,  
for $t>0$, one has 
\begin{align} \label{Y-7}
\begin{aligned}
\prod_{i<j} \Big(1-R_{ij}(t) \Big)  =  \prod_{i<j} \Big( 1-R_{ij}^0 \Big), \quad \textup{or equivalently,} \quad  \prod_{i<j} \|x_i(t)-x_j(t)\| = \prod_{i<j} \|x_i^0-x_j^0\|.
\end{aligned}
\end{align}
\end{theorem}
\begin{proof} 
Without loss of generality assume that $x_i\neq x_j$ for all $i,j$. Then $R_{ij}:=\langle x_i,x_j\rangle$ satisfies 
\begin{equation} \label{Y-8}
\frac{d}{dt} \langle x_i,x_j\rangle = \frac{\kp}{N}  (1-\langle x_i,x_j\rangle ) \Big[ \sum_{k\neq i} \langle x_i,Wx_k\rangle + \sum_{k\neq j } \langle x_j,Wx_k\rangle \Big].
\end{equation}
Since $W$ is a skew-symmetric matrix, we have
\begin{align*}
\sum_{k\neq j} \langle x_j, Wx_k\rangle  = \sum_{k\neq j} \langle -Wx_j,x_k\rangle = \sum_{k\neq j} \langle x_k,-Wx_j\rangle = -\sum_{k\neq j} \langle x_j,Wx_k\rangle.
\end{align*}
We sum up \eqref{Y-8} with respect to the indices $i <j$ and divide the resulting relation by $1-\langle x_i,x_j\rangle$  to get
\begin{align*}
\frac{d}{dt} \sum_{i <j} \frac{ \langle x_i,x_j\rangle}{1- \langle x_i,x_j\rangle} &= \frac{\kp}{N} \sum_{i <j} \Big[ \sum_{k\neq i} \langle x_i,Wx_k\rangle + \sum_{k\neq j } \langle x_j,Wx_k\rangle \Big] \\
&= \frac{\kp}{N} \sum_{i<j} \Big[ \sum_{k\neq i} \langle x_i,Wx_k\rangle -\sum_{k\neq j} \langle x_j,Wx_k\rangle\Big] \\
&=0.
\end{align*}
We integrate the above relation with respect to time $t$ to obtain the desired result.
\end{proof}
\begin{remark} \label{R5.5}
As  discussed before in Remark \ref{RY.1}, if the particles are located in all distinct positions initially, i.e., $\langle x_i^0,x_j^0\rangle \neq 1$, for all $i,j$, then the identity \eqref{Y-7} in the Theorem \ref{T3.2} yields that any two particles cannot converge towards each other. Hence, we may conclude that the skew-symmetric part of the frustration matrix contributes to anti-synchronous behavior. 

%\noindent 2. Before we close this section, we bring up an interesting problem in the following remark. \textcolor{blue}{The Lohe sphere model `without' frustration conserves the momentum so that after the particles aggregate},\textcolor{black}{why is this true?} every particles have to stop at the end. On the other hand for the Lohe sphere model `with' frustration, momentum is not conserved so that even though all particles aggregate, they can still move forming one group. Then we address the question about the trajectory on which the aggregated particles follow, more precisely, we just wonder whether the aggregated particles are moving around on the great circle, or not.
\end{remark}

\section{The Lohe matrix model with frustration}  \label{sec:6}
\setcounter{equation}{0}
In this section, we study emergent dynamics and equilibria of the Lohe matrix model with frustration for identical hamiltonians $D(H) = 0$. 
\subsection{Complete aggregation} \label{sec:6.1} Consider the Lohe matrix model with identical hamiltonians:
\begin{equation} \label{F-0}
\mi \dot U_j U_j ^* = H + \frac{\mi \kappa}{2N} \sum_{k=1}^N ( VU_kU_j^* - U_j U_k^* V^*), \quad j = 1, \cdots, N,
\end{equation}
where the frustration matrix $V$ is given to be unitary, that is, $V\in\mathbb{U}(d)$.   It is worthwhile to mention the recent results \cite{Lo-6} which concerns the constants of motion for \eqref{F-0} without frustration $V=I_d$. To be precise, they defined the matrix cross-ratios:
\begin{equation*}
C_{ijk\ell} := (U_i  - U_k)(U_i - U_\ell)^{-1}(U_j - U_\ell)(U_j - U_k)^{-1},\quad  i\neq \ell , ~~  j\neq k . 
\end{equation*}
Then, time-evolution of $C_{ijk\ell}$ is given by 
\begin{equation*}
\dot C_{ijk\ell} = \frac\kp2[ C_{ijk\ell}, U_kU_c^\dg + \mi H], \quad U_c:=\frac1N\sum_{i=1}^N U_i,
\end{equation*}
where $[\cdot,\cdot]$ denotes the usual commutator of two matrices. Then, the eigenvalues become constants of motion. For the detailed argument and proof, we refer the reader to \cite{Lo-6} and references therein. \newline

Now, we present the first main result on the exponential  aggregation of the identical oscillators. Note that $\|A\|_\textup{F}$ denotes the Frobenius norm for a $d\times d$ matrix $A$. 

\begin{theorem} \label{T6.1}
Suppose that the frustration matrix and the initial data satisfy
\begin{equation} \label{F-0-0-0}
 \|V-I_d\|_\textup{F} <\frac{2}{3} \quad \textup{and} \quad \max_{1\leq i,j\leq N} \|U_i^0-U_j^0\|_\textup{F} <\sqrt{2-3\|V-I_d\|_\textup{F}},  
\end{equation} 
 and let $\{U_i(t)\}$ be a solution to \eqref{F-0}. Then complete aggregation emerges asymptotically: 
\[
\lim_{t\to\infty} \max_{1\leq i,j\leq N} \|U_i(t) - U_j(t)\|_\textup{F} = 0.
\]
\end{theorem}
\begin{proof}  For $i, j  = 1, \cdots, N$,  we set 
\begin{equation} \label{F-0-1}
  G_{ij}:=U_iU_j^*, \quad  L_{ij}:=I_d-G_{ij}. 
\end{equation} 
In order to derive the dynamics of $L_{ij}$, we first estimate the time-evolution of $G_{ij}$, and then we derive the estimate for $L_{ij}$. \newline

\noindent $\bullet$~Step A~(Estimate of $G_{ij}$): We use equations for $U_i$ and $U_j$:
\begin{align}
\begin{aligned} \label{X-2}
\dot U_i &= -\mi H U_i + \frac{\kp}{2N} \sum_{k=1}^N (VU_k - U_iU_k^* V^* U_i), \\
\dot U_j^* &= \mi U_j^* H + \frac{\kp}{2N} \sum_{k=1}^N ( U_k^* V^* - U_j^* V U_k U_j^*) 
\end{aligned}
\end{align}
to get
\begin{align} \label{X-3}
\begin{aligned}
&\dot U_i U_j^* = -\mi H U_iU_j^* + \frac{\kappa}{2N} \sum_{k=1}^N (VU_kU_j^* - U_iU_k^* V^* U_iU_j^*), \\
& U_i\dot U_j^* = \mi U_iU_j^* H + \frac{\kp}{2N} \sum_{k=1}^N ( U_iU_k^* V^* - U_iU_j^* V U_k U_j^*).
\end{aligned}
\end{align}
Now, we add \eqref{X-2} and \eqref{X-3} to see 
\begin{align*}
\begin{aligned}
\frac{d}{dt} (U_iU_j^*) &= \mi ( U_iU_j^* H - H U_i U_j^* )  \\
&\hspace{0.2cm} + \frac{\kp}{2N} \sum_{k=1}^N (VU_kU_j^* - U_iU_k^* V^* U_iU_j^* + U_iU_k^* V^* - U_iU_j^* VU_kU_j^*).
\end{aligned}
\end{align*}
or equivalently, $G_{ij}$ satisfies
\begin{equation} \label{F-1}
\frac{d}{dt} G_{ij} = \mi(G_{ij}H-HG_{ij}) + \frac{\kp}{2N}\sum_{k=1}^N ( VG_{kj} - G_{ik}V^* G_{ij} + G_{ik}V^*- G_{ij}VG_{kj}).
\end{equation}

\vspace{0.2cm}

\noindent $\bullet$~Step B~(Estimate of $L_{ij}$): We use \eqref{F-1} to find the dynamics of $I_d-G_{ij}$ 
\begin{align} \label{F-1-0}
\begin{aligned}
&\frac{d}{dt} (I_d-G_{ij}) \\
& \hspace{0.2cm} = \mi(I_d-G_{ij})H - \mi H(I_d-G_{ij}) + \frac{\kp}{2N} \sum_{k=1}^N ( -G_{ik}V^*(I_d-G_{ij}) - (I_d-G_{ij})VG_{kj} )\\
&\hspace{0.2cm}  = \mi(I_d-G_{ij})H - \mi H(I_d-G_{ij}) \\
&\hspace{0.4cm} +\frac{\kp}{2N} \sum_{k=1}^N \Big[ (I_d-G_{ij})V(I_d-G_{kj})- (I_d-G_{ij})V +(I_d-G_{ik})V^*(I_d-G_{ij}) - V^*(I_d-G_{ij})\Big],
\end{aligned}
\end{align}
or equivalently, \eqref{F-1-0} can be rewritten in terms of  $L_{ij}$ 
\begin{equation} \label{X-4}
\frac{d}{dt} L_{ij} = \mi (L_{ij}H-HL_{ij}) + \frac{\kp}{2N} \sum_{k=1}^N ( L_{ij}VL_{kj} + L_{ik}V^* L_{ij}) -\frac{\kp}{2} ( L_{ij}V+V^* L_{ij}).
\end{equation}
We use the relations \eqref{F-0-1} to get
\[ \|L_{ij}\|_\textup{F}^2 = \textup{tr}[ L_{ij}L_{ji} ] = \textup{tr}[L_{ij} + L_{ji}].\]
Then, the combination $\eqref{X-4} + \eqref{X-4}^*$ yields
\begin{align*}
\begin{aligned} \label{X-5}
\frac{d}{dt} \|L_{ij}\|_\textup{F}^2 &= -\frac{\kp}{2} \textup{tr}[L_{ij}V + V^* L_{ij} + V^* L_{ij}^* + L_{ij}^* V ]\\
&\hspace{0.5cm}+\frac{ \kp}{2N}\sum_{k=1}^N \textup{tr}[L_{ij}VL_{kj} + L_{ik}V^* L_{ij} +(L_{ij}VL_{kj})^* + (L_{ik}V^* L_{ij})^*] \\
&=: -\frac{\kp}{2}{ \mathcal I}_{11} + {\mathcal I}_{12}.
\end{aligned}
\end{align*}
Below, we estimate the terms ${\mathcal I}_{11}$ and $\mathcal I_{12}$, separately. \newline

\noindent $\bullet$~(Estimate of ${\mathcal I}_{11}$): By direct estimate, one has
\begin{align*}
\begin{aligned}
{\mathcal I}_{11} = \textup{tr}[ V^* L_{ij}L_{ji} + L_{ij}L_{ji}V] =\textup{tr}[ L_{ij}L_{ji}(V+V^*)] = 2\|L_{ij}\|_\textup{F}^2+ \textup{tr}[ L_{ij}L_{ji}(V+V^*-2I_d)].
\end{aligned}
\end{align*}
This yields
\begin{equation*} \label{X-6}
 |{\mathcal I}_{11} -2\|L_{ij}\|_\textup{F}^2| \le 2 \|L_{ij}\|_\textup{F}^2 \|V-I_d\|_\textup{F}. 
 \end{equation*}
\noindent $\bullet$~(Estimate of ${\mathcal I}_{12}$): Similarly, we find
\begin{align*}
\begin{aligned}
&\textup{tr}[L_{ij}VL_{kj} + L_{ik}V^* L_{ij} +(L_{ij}VL_{kj})^* + (L_{ik}V^* L_{ij})^*]\\
& \hspace{0.5cm} = \textup{tr} [L_{ij}VL_{kj} + L_{ik}V^* L_{ij} +L_{jk} V^* L_{ji} + L_{ji}V L_{ki}]\\
& \hspace{0.5cm} = \textup{tr}[L_{ij}L_{kj} + L_{ik} L_{ij} +L_{jk} L_{ji} + L_{ji} L_{ki}]\\
& \hspace{0.7cm} + \textup{tr}[L_{ij}(V-I_d)L_{kj} + L_{ik}(V^*-I_d) L_{ij} +L_{jk} (V^*-I_d) L_{ji} + L_{ji}(V-I_d) L_{ki}] \\
& \hspace{0.5cm} =: {\mathcal I}_{121} +  {\mathcal I}_{122}. 
\end{aligned}
\end{align*}

\noindent $\diamond$~(Estimate of ${\mathcal I}_{122}$): By direct estimates,
\[ |{\mathcal I}_{122}| \leq 4D(U)^2 \|V-I_d\|_\textup{F}. \]

\noindent $\diamond$~(Estimate of ${\mathcal I}_{121}$): Similarly, one has
\begin{align*}
\begin{aligned}
&\textup{tr}[L_{ij}L_{kj} + L_{ik} L_{ij} +L_{jk} L_{ji} + L_{ji} L_{ki}]\\
& \hspace{0.2cm} =\textup{tr}[(I_d-U_i U_j^* -U_k U_j^* +U_i U_j^* U_k U_j^*)+(I_d -U_i U_k^*-U_iU_j^* +U_iU_k^* U_iU_j^*)\\
&\hspace{0.4cm} +(I_d-U_j U_k^*-U_j U_i^*+U_jU_k^* U_j U_i^*)+(I_d-U_j U_i^*-U_k U_i^* +U_j U_i^* U_k U_i^*)]\\
& \hspace{0.2cm} =\textup{tr}[4I_d-2U_i U_j^*-2U_j U_i^* -U_j U_k^*-U_k U_j^*-U_i U_k^*-U_k U_i^*+U_j U_i^* U_j U_k^*+U_i U_j^* U_k U_j^* ].
\end{aligned}
\end{align*}
On the other hand, note that 
\begin{align*}
\begin{aligned}
&-2\|L_{ij}\|^2+ \textup{tr}[L_{ij}L_{ji}L_{jk}L_{kj}+L_{ji}L_{ij}L_{ki}L_{ik}]\\
&=\textup{tr}[-2(2I_d-U_i U_j^*-U_j U_i^*)+(2I_d-U_i U_j^*-U_j U_i^*)(4I_d-U_j U_k^*-U_k U_j^*-U_i U_k^*-U_k U_i^*)]\\
&=\textup{tr}[(2I_d-U_i U_j^*-U_j U_i^*)(2I_d-U_j U_k^*-U_k U_j^*-U_i U_k^*-U_k U_i^*)]\\
&=\textup{tr}[(4I_d-2U_i U_j^*-2U_j U_i^*) -(2U_j U_k^*-U_i U_k^*-U_j U_i^* U_j U_k^*) \\
&\quad -(2U_k U_j^*-U_i U_j^* U_k U_j^*-U_j U_i^* U_k U_j^*)-(2U_i U_k^*-U_i U_j^* U_i U_k^*-U_j U_k^*)\\
&\quad-(2U_k U_i^*-U_i U_j^* U_k U_i^*-U_j U_i^* U_k U_i^*)]\\
&=\textup{tr}[(4I_d-2U_i U_j^*-2U_j U_i^*) -(2U_j U_k^*-U_i U_k^*-U_j U_i^* U_j U_k^*) \\
&\quad-(2U_k U_j^*-U_i U_j^* U_k U_j^*- U_i^* U_k) -(2U_i U_k^*-U_i U_j^* U_i U_k^*-U_j U_k^*) \\
&\quad -(2U_k U_i^*- U_j^* U_k -U_j U_i^* U_k U_i^*)]\\
&=\textup{tr}[4I_d-2U_i U_j^*-2U_j U_i^* -2U_j U_k^*+U_i U_k^* +U_j U_i^* U_j U_k^*-2U_k U_j^*+U_i U_j^* U_k U_j^* + U_i^* U_k\\
&\quad -2U_i U_k^*+U_i U_j^* U_i U_k^* +U_j U_k^*-2U_k U_i^*+ U_j^* U_k +U_j U_i^* U_k U_i^*)]\\
&=\textup{tr}[4I_d-2U_i U_j^*-2U_j U_i^* -U_j U_k^*-U_k U_j^*-U_i U_k^*-U_k U_i^*+U_j U_i^* U_j U_k^*+U_i U_j^* U_k U_j^* \\
&\quad +U_i U_j^* U_i U_k^* +U_j U_i^* U_k U_i^*].
\end{aligned}
\end{align*}

\noindent $\bullet$~Step C~(Derivation of Gr\"onwall's inequality):  We now define the maximal diameter as
\begin{equation*}
D(U):= \max_{1\leq i,j\leq N} \|U_i - U_j\|_\textup{F} = \max_{1\leq i,j\leq N} \|I_d - U_iU_j^*\|_\textup{F} = \max_{1\leq i,j\leq N}\| L_{ij}\|_\textup{F}.
\end{equation*}
For each $t>0$, we choose extremal indices $(i_t,j_t)$ satisfying 
\begin{equation*}
D(U) = \|U_{i_t} - U_{j_t}\|_\textup{F} = \|L_{i_tj_t}\|_\textup{F}.
\end{equation*}
Hence, we obtain 
\[
\left|\frac{d}{dt} D(U)^2+2\kappa D(U)^2 \right|\le \kappa D(U)^4+3\kappa D(U)^2 \|V-I_d\|_\textup{F},
\]
or  equivalently, 
\begin{equation*}
\frac{d}{dt} D(U) \leq -\frac{\kp}{2} \Big (2 -3\|V-I_d\|_\textup{F} \Big)D(U) + \frac{\kp}{2} D(U)^3.
\end{equation*}
Since the initial data satisfy the condition \eqref{F-0-0-0}, we obtain the desired result. Moreover the solution of the Riccati-type differential inequality and comparison principle yield exponential aggregation. 
\end{proof}

\begin{remark}\label{R6.1}
(i)~If we assume that there is no frustration, i.e., $V=I_d$, then the conditions in \eqref{F-0-0-0} reduces to 
\begin{equation*}
D(U^0) < \sqrt2
\end{equation*}
which is already introduced in \cite{H-R}. Moreover, the smallness condition in \eqref{F-0-0-0} says that the frustration matrix must be not much different from the identity matrix.

(ii) Although Theorem \ref{T6.1} is stated in terms of the Frobenius norm $\|\cdot\|_F$, a similar result holds(with a similar proof) for the operator norm $\|\cdot \|_{op}$:
\[
\|V-I_d\|_{op}<\frac{1}{4}\quad \mbox{and}\quad \max_{1\le i,j\le N}\|U_i^0-U_j^0\|_{op}<1
\]
implies
\[
\lim_{t\rightarrow\infty}\max_{1\le i,j\le N}\|U_i(t)-U_j(t)\|_{op}=0.
\]
The constants are not optimal and are in the process of improvement. This result will be published in a forthcoming manuscript.
\end{remark}

\subsection{Equilibria} \label{sec:6.2}
Recall that for the Kuramoto model with identical oscillators, some equilibria with order parameter zero are the splay states, that is, the states where the initial phases are spaced equally on the unit circle and fixed for all time $t$. These states correspond to the finite subgroups of $\bbs^1$, or equivalently, the image of group homomorphisms from finite groups into the circle group $\bbs^1$. Of course, these do not fully characterize the equilibrium states with order parameter zero, and it is easy to construct configurations that possess no plane symmetry. However, we may generalize this idea to the Lohe matrix model by considering embeddings of finite groups into unitary groups of higher dimension. Consider the Lohe matrix model for $H_j = 0$:
\begin{equation} \label{F-2}
\dot U_j  U_j^* =  \frac{\kappa}{2N} \sum_{k=1}^N ( VU_kU_j^* - U_j U_k^* V^*), \quad j = 1, \cdots, N.
\end{equation}

\begin{theorem} \label{T6.2}
Let $G=\{g_1,\cdots,g_N\}$ be a finite group.
\begin{enumerate}
\item If $\varrho:G\rightarrow {\bbu}(d)$ is a group homomorphism, then the initial condition
\[
N=|G|,\quad U_i^0=\varrho(g_i),~i=1,\cdots,N
\]
is an equilibrium state for the  Lohe matrix model \eqref{F-2} with $V=I_d$.

\vspace{0.2cm}

\item Let $\varrho:G\rightarrow {\bbu}(d)$ be an irreducible unitary representation of $G$. Then the initial condition
\[
N=|G|,\quad U_i^0=\varrho(g_i),~i=1,\cdots,N
\]
is an equilibrium state for the Lohe matrix model \eqref{F-2}.
\end{enumerate}
\end{theorem}
\begin{proof}
\textcolor{black}{
\noindent (i)~ For fixed $g_i \in G$, it follows from $|G|=N<\infty$ that 
\begin{equation*}
\{ g_1,\cdots, g_N \} = \{ g_1g_i^{-1}, \cdots, g_Ng_i^{-1}\} =\{ g_ig_1,\cdots, g_ig_N\}= \{ g_1^{-1}, \cdots, g_N^{-1}\}.
\end{equation*}
Thus, we see 
\begin{equation*}
\sum_{k=1}^N \varrho(g_k g_i^{-1}) = \sum_{k=1}^N \varrho(g_k) , \quad \sum_{k=1}^N \varrho(g_ig_k^{-1}) = \sum_{k=1}^N \varrho(g_k^{-1})  = \sum_{k=1}^N \varrho(g_k).
\end{equation*}
Thus, we obtain the desired result. 
\[
\sum_{k=1}^{N}  (U_k^0 (U_i^0)^* -U_i^0 (U_k^0)^*)=\sum_{k=1}^{N}  (\varrho(g_k g_i^{-1})-\varrho(g_i g_k^{-1})) =\sum_{k=1}^{N} (\varrho(g_k)-\varrho(g_k))=0.
\]
\noindent (ii)~We set
\[
U^0_c= \frac{1}{N} \sum_{k=1}^{N} \varrho(g_k).
\]
Then for all $i=1,\cdots,N$,
\[
\varrho(g_i)U_c^0= \frac{1}{N} \varrho(g_i)\sum_{k=1}^{N} \varrho(g_k)= \frac{1}{N} \sum_{k=1}^{N} \varrho(g_i g_k)= \frac{1}{N} \sum_{k=1}^{N} \varrho(g_k)=U_c^0.
\]
Hence the image of $U_c^0$ becomes a common eigenspace of all the $\varrho(g_i)$'s. By irreducibility, the image space of $U_c^0$ must be either zero or the whole space. \newline
\noindent $\bullet$~Case A (the image space of $U_c^0$ is zero): We use  $U_c^0=0$ to see
\[
\sum_{k=1}^{N} (VU_k^0 (U_i^0)^*-U_i^0 (U_k^0)^* V^*)=V U_c^0 (U_i^0)^*-U_i^0 (U_c^0)^* V^*=0.
\]
\noindent $\bullet$~Case B (the image space of $U_c^0$ is the whole space): In this case, all the $U_i^0$'s act as the identity on the whole space, i.e. $\varrho$ is the trivial representation. Thus, $\varrho(g_i)$'s then correspond to the completely aggregated state.}
\end{proof}
\begin{remark}
It is well-known that any finite group has an irreducible unitary representation. By considering that the structure as well as the representation of finite groups is vast, we can say that some of the splay states in the higher dimensional unitary groups are complex yet highly symmetric. As a concrete example, let us consider the standard representation $\varrho$ of the symmetric group $S_n$. It is an irreducible representation of dimension $n-1$, and is the symmetry group of the $(n-1)$-dimensional standard simplex. In this case, it is easy to see that
\[
D(\mbox{im}~\varrho)=\sqrt{2n},\quad D_{op}(\mbox{im}~\varrho)=\sqrt{\frac{2n}{n-1}}.
\]
(Here $D_{op}$ is defined the same as $D$ with the Frobenius norm replaced by the operator norm.) This tells us that the forthcoming result mentioned in Remark \ref{R6.1} (ii) is close to being sharp and has the necessary limitation $D(U^0)<\sqrt{2}$.
\end{remark}

\section{conclusion} \label{sec:7}
\setcounter{equation}{0}
In this paper, we have investigated constants of motion and low-dimensional reductions of the Lohe type models under the frustration effect. We have also presented asymptotic aggregation estimates which provides frameworks in terms of the initial data and frustrations. As we already have seen in the previous literature, frustration brings in a competition between `synchronous motion (aggregation)' and `periodic motion'. This coincides with the Lohe sphere model case where the `synchronous motion' and `periodic motion' are represented by the identity matrix and the skew-symmetric matrix, respectively. We did not address the emergence of periodic motions for non-identical particles, and we leave it to future work. We also provide a unified framework for the constants of motion for the Kuramoto model and Lohe sphere model with frustration and suggest applications to asymptotic behaviors of the Lohe type models. We further present a method of reduction with which we can reduce the number of effective variables. Finally, we provide some class of equilibria of the Lohe matrix model with identical hamiltonians.

\newpage

\appendix

\section{Proof of Proposition \ref{P3.2}} \label{App-A}
\setcounter{equation}{0}
In this appendix, we provide a detailed proof of Proposition \ref{P3.2}. For this, we first show that the Cauchy problem of Proposition \ref{P3.2} admits the unique global-in-time smooth solution for $f(t)$ and $g(t)$, and then the relation \eqref{affine} holds:
\[
x_j(t) = g(t)+f(t)x_j^0,\quad j=1,\dots, N-m.
\]

\subsection{Well-posedness of the Cauchy problem} Recall  the Cauchy problem:
\[
\begin{pmatrix}
f'(t) \\ g'(t)
\end{pmatrix}
=
\begin{pmatrix}
\tilde{\mathcal{B}}f(t) \\ \tilde{\mathcal{A}}+\tilde{\mathcal{B}}g(t)
\end{pmatrix}
,
\quad
\begin{pmatrix}
f(0) \\ g(0)
\end{pmatrix}
=
\begin{pmatrix}
1 \\ 0
\end{pmatrix},
\]
where
\[
\tilde{\mathcal{A}}=\frac{\kappa}{N}\left[m\sin\alpha+\sum_{k=1}^{N-m}\left(\frac{2f(t)x_k^0+2g(t)}{(f(t)x_k^0+g(t))^2+1}\cos\alpha+\frac{(f(t)x_k^0+g(t))^2-1}{(f(t)x_k^0+g(t))^2+1}\sin\alpha\right)\right]
\]
and
\[
\tilde{\mathcal{B}}=\frac{\kappa}{N}\left[m\cos\alpha+\sum_{k=1}^{N-m}\left(-\frac{2(f(t)x_k^0+g(t))}{(f(t)x_k^0+g(t))^2+1}\sin\alpha+\frac{(f(t)x_k^0+g(t))^2-1}{(f(t)x_k^0+g(t))^2+1}\cos\alpha\right)\right],
\]
for given real numbers $x_1^0,\cdots, x_{N-m}^0\in\bbr$.

First of all, $\tilde{\mathcal{A}}$ and $\tilde{\mathcal{B}}$ are rational functions of $f$ and $g$ and have no poles on the real line, so this regularity easily implies the local existence and uniqueness of $f(t)$ and $g(t)$. So it remains to put an upper bound to the growth rate of $f(t)$ and $g(t)$ so as to exclude finite-time blow-up. Note that 
\[
\left(\frac{2y}{y^2+1}\right)^2+\left(\frac{y^2-1}{y^2+1}\right)^2=1,\quad \forall y\in \bbr.
\]
By the Cauchy-Schwarz inequality, one has
\begin{align*}
\begin{aligned}
|\tilde{\mathcal{A}}| &= \frac{\kappa}{N}\left[m|\sin\alpha|+\sum_{k=1}^{N-m}\left(\left|\frac{2f(t)x_k^0+2g(t)}{(f(t)x_k^0+g(t))^2+1}\right| |\cos\alpha|+\left|\frac{(f(t)x_k^0+g(t))^2-1}{(f(t)x_k^0+g(t))^2+1}\right| |\sin\alpha|\right)\right]\\
&\le \frac{\kappa}{N}\left[m+\sum_{k=1}^{N-m}\left(\left|\frac{2(f(t)x_k^0+g(t))}{(f(t)x_k^0+g(t))^2+1}\right|^2+\left|\frac{(f(t)x_k^0+g(t))^2-1}{(f(t)x_k^0+g(t))^2+1}\right|^2\right)^{\frac{1}{2}}\left(|\cos\alpha|^2+|\sin\alpha|^2\right)^{\frac{1}{2}}\right]\\
&\le \frac{\kappa}{N}\left[m+(N-m)\right] =\kappa.
\end{aligned}
\end{align*}
Similarly, we have
\begin{align*}
|\tilde{\mathcal{B}}|\le & \frac{\kappa}{N}\left[m|\cos\alpha|+\sum_{k=1}^{N-m}\left(\left|\frac{2(f(t)x_k^0+g(t))}{(f(t)x_k^0+g(t))^2+1}\right||\sin\alpha|+\left|\frac{(f(t)x_k^0+g(t))^2-1}{(f(t)x_k^0+g(t))^2+1}\right||\cos\alpha|\right)\right]\\
\le &\frac{\kappa}{N}\left[m+\sum_{k=1}^{N-m}\left(\left|\frac{2(f(t)x_k^0+g(t))}{(f(t)x_k^0+g(t))^2+1}\right|^2+\left|\frac{(f(t)x_k^0+g(t))^2-1}{(f(t)x_k^0+g(t))^2+1}\right|^2\right)^{\frac{1}{2}}\left(|\sin\alpha|^2+|\cos\alpha|^2\right)^{\frac{1}{2}}\right]\\
\le &\frac{\kappa}{N}\left[m+(N-m)\right] = \kappa.
\end{align*}
Therefore, one finds
\[
\begin{cases}
|f'(t)|\le |\tilde{\mathcal{B}}f(t)|\le \kappa |f(t)|,\quad |f(0)|=1, \\
|g'(t)| \le |\tilde{\mathcal{A}}|+|\tilde{\mathcal{B}}||g(t)|\le \kappa + \kappa |g(t)|,\quad |g(0)|=0.
\end{cases}
\]
These imply
\[
|f(t)|\le e^{\kappa t},\quad |g(t)|\le e^{\kappa t}-1,\quad t>0. 
\]
Thus, $f$ and $g$ must be defined for all $t\ge 0$.

\subsection{Verification of \eqref{affine}} We define
\[
\tilde{x}_j(t)=g(t)+f(t)x_j^0,\quad j=1,\cdots,N-m, \quad t\ge 0.
\]
For $j=1,\cdots,N-m$, one has
\[
\dot{\tilde{x}}_j(t) =g'(t)+f'(t)x_j^0 =\tilde{\mathcal{A}}+\tilde{\mathcal{B}}g(t)+\tilde{\mathcal{B}}f(t)x_j^0 =\tilde{\mathcal{A}}+\tilde{\mathcal{B}}\tilde{x}_j.
\]
Furthermore, we have
\[
\tilde{\mathcal{A}}=\mathcal{A}(\{\tilde{x}_j\}),\quad \tilde{\mathcal{B}}=\mathcal{B}(\{\tilde{x}_j\}),
\]
so $\{\tilde{x}_j \}_{j=1}^{N-m}$ is a solution of the Cauchy problem \eqref{stereo}:
\[
\begin{cases}
\dot{\tilde{x}}_j =\mathcal{A}(\{\tilde{x}_j \})+\tilde{\mathcal{B}}(\{\tilde{x}_j \})\tilde{x}_j,\\
\displaystyle \tilde{x}_j(0)=x_j(0)=\frac{1+\cos(\theta_j^0-\theta_N^0)}{\sin(\theta_j^0-\theta_N^0)},\quad j = 1,\cdots, N-m.
\end{cases}
\]
However, $\{x_j \}_{j=1}^{N-m}$ is also a solution of the Cauchy problem \eqref{stereo}, as demonstrated by Proposition \ref{P3.2}. Since the Cauchy problem \eqref{stereo} has clearly a unique solution, we can conclude
\[
\tilde{x}_j(t)=x_j(t),\quad j=1,\cdots,N-m,\quad t\ge 0,
\]
which is an equivalent statement of \eqref{affine}.

\end{document}